\documentclass[11pt]{article}

\title{A Unified Approach for Clustering Problems on Sliding Windows}

\author{
Vladimir Braverman\thanks{Department of Computer Science, Johns Hopkins University. {\tt vova@cs.jhu.edu},
This material is based upon work supported in part by the
National Science Foundation under Grant No. 1447639, by the Google
Faculty Award and by DARPA grant N660001-1-2-4014. Its contents are
solely the responsibility of the authors and do not represent the
official view of DARPA or the Department of Defense.}
\and
Harry Lang\thanks{ Department of Mathematics, Johns Hopkins University. {\tt hlang8@jhu.edu}}
\and
Keith Levin\thanks{ Department of Computer Science, Johns Hopkins University. {\tt klevin@jhu.edu}}
\and 
Morteza Monemizadeh\thanks{Computer Science Institute of Charles University,
Faculty of Mathematics and Physics, Prague, Czech Republic. Partially supported by the project 14-10003S of GA \v{C}R. Part of this work was done when the author was at Department of Computer Science, Goethe-Universit\"{a}t Frankfurt, Germany and supported in part by MO 2200/1-1. {\tt monemi@iuuk.mff.cuni.cz}}
}

\usepackage{makeidx}
\usepackage{euscript}
\usepackage{amsmath,amssymb,amsfonts}
\usepackage{thmtools}
\usepackage{thm-restate}
\usepackage{dsfont}
\usepackage{latexsym}
\usepackage{subfigure}
\usepackage{graphicx}
\usepackage{times}\usepackage[scaled=0.92]{helvet} 
\usepackage{euler}
\usepackage{fancybox}
\usepackage{fullpage}
\usepackage{hyperref}
\usepackage{paralist}

\usepackage{color}
\usepackage{wrapfig}
\usepackage{tikz}

\usepackage[small,compact]{titlesec}

\usepackage[noend]{algorithmic}
\usepackage{algorithm}

\newcount\shortyear\newcount\shorthour\newcount\shortminute
\shorthour=\time\divide\shorthour by 60\shortyear=\shorthour
\multiply\shortyear by 60\shortminute=\time\advance\shortminute by
-\shortyear \shortyear=\year\advance\shortyear by -1900

\def\zeit{\number\shorthour:\ifnum\shortminute<10 0\number\shortminute
\else\number\shortminute\fi}

\newenvironment{proof}{\noindent{\bf Proof : \ }}{\hfill$\Box$\par\medskip}

\newtheorem{theorem}{Theorem}

\newtheorem{lemma}[theorem]{Lemma}

\newtheorem{definition}[theorem]{Definition}

\newtheorem{claim}[theorem]{Claim}

\newenvironment{proofof}[1]{\begin{trivlist} \item {\bf Proof
#1:~~}}
  {\qed\end{trivlist}}
\renewenvironment{proofof}[1]{\par\medskip\noindent{\bf Proof of #1: \ }}{\hfill$\Box$\par\medskip}

\renewcommand{\paragraph}[1]{\medskip \noindent {\bf #1}}

\newcommand{\COMMENTED}[1]{{}}

\newcommand{\NAT}{\ensuremath{\mathbb{N}}}
\newcommand{\NATURAL}{\NAT}
\newcommand{\REAL}{\ensuremath{\mathbb{R}}}

\newcommand{\poly}{{\mathrm{poly}}}

\newcommand{\tri}{\lambda} 
\newcommand{\PLS}{\operatorname{PLS}}
\newcommand{\textPLS}{{\sc PLS }}
\newcommand{\textPLSns}{{\sc PLS}}
\newcommand{\store}{\operatorname{store}}
\newcommand{\KHmark}{\operatorname{mark}}

\newcommand{\update}{{\sc Update}}

\newcommand{\cost}{\ensuremath{\text{\footnotesize\textsf{COST}}}}
\newcommand{\opt}{\ensuremath{\text{\footnotesize\textsf{OPT}}}}
\newcommand{\dist}{\text{dist}}
\renewcommand{\dist}{\ensuremath{\text{\footnotesize\textsf{DIST}}}}
\renewcommand{\dist}{\ensuremath{d}}

\newlength{\savedparindent}

\begin{document}

\sloppy
\begin{titlepage}
\maketitle\thispagestyle{empty}

\begin{abstract}
We explore clustering problems in the streaming sliding window model
in both general metric spaces and Euclidean space.
We present the first polylogarithmic space $O(1)$-approximation to
the metric $k$-median and metric $k$-means problems in the sliding window model,
answering the main open problem posed by
Babcock, Datar, Motwani and O'Callaghan~\cite{BDMO03},
which has remained unanswered for over a decade.
Our algorithm uses $O(k^3 \log^6 n)$ space
and $\poly(k, \log n)$ update time.
This is an exponential improvement on the space required by the
technique due to Babcock, et al.
We introduce a data structure that extends
smooth histograms as introduced by Braverman and Ostrovsky~\cite{BO06}
to operate on a broader class of functions.
In particular, we show that using only polylogarithmic space
we can maintain a summary of the current window from which we
can construct an $O(1)$-approximate clustering solution.

Merge-and-reduce is a generic method in computational geometry for
adapting offline algorithms to the insertion-only streaming model.
Several well-known coreset constructions are maintainable in the
insertion-only streaming model using this method,
including well-known coreset techniques for the $k$-median, $k$-means 
in both low-and high-dimensional
Euclidean spaces~\cite{HPM04,Ch09}.
Previous work has adapted these techniques to the insertion-deletion
model, but translating them to the sliding window model has remained
a challenge.
We give the first algorithm that, given an insertion-only
streaming coreset construction of space $s$,
maintains a $(1\pm\epsilon)$-approximate coreset in the sliding window model
using $O(s^2\epsilon^{-2}\log n)$ space.

For clustering problems, our results constitute the first significant
step towards resolving problem number 20 from the List of Open Problems
in Sublinear Algorithms~\cite{sublinear_open_20}.

\end{abstract}
\end{titlepage}
\maketitle

\section{Introduction}
\label{sec:intro}
Over the past two decades, the streaming model of computation~\cite{M05}
has emerged as a popular framework in which to develop algorithms for
large data sets.
In the streaming model, we are restricted to using space sublinear in the size
of the input, and this input must typically must be processed in a single pass.
While the streaming model is broadly useful,
it is inadequate for domains in which data is time-sensitive
such as network monitoring~\cite{C13,CG08,CM05x}
and event detection in social media~\cite{OsborneEtAl2014}.
In these domains, elements of the stream appearing more recently are
in some sense more relevant to the computation being performed.
The sliding window model was developed to capture this
situation~\cite{DGIM02}.
In this model, the goal is to maintain a computation on only the most
recent $W$ elements of the stream, rather than on the stream in
its entirety.

We consider the problem of clustering in the sliding window model.
Algorithms have been developed for a number of streaming clustering
problems, including $k$-median~\cite{GMMO00,COP03, HPM04,FS05},
$k$-means~\cite{Ch09,FMS07} and facility location~\cite{CLMS13}.
However, while the sliding window model has received renewed attention
recently~\cite{CMS13,BCM13},
no major clustering results in this model have been published since
Babcock, Datar, Motwani and O'Callaghan~\cite{BDMO03}
presented a solution to the $k$-median problem.
Polylogarithmic space $k$-median algorithms exist
in the insertion-only streaming model~\cite{COP03,HPM04} and the
insertion-deletion model~\cite{Ind04,FS05,IP11},
but no analogous result has appeared to date for the sliding window model.
Indeed, the following question by Babcock, et al.~\cite{BDMO03}
has remained open for more than a decade:
\begin{quote} \itshape
Whether it is possible to maintain approximately optimal
medians in polylogarithmic space (as Charikar et al.~\cite{COP03} do
in the stream model without sliding windows), rather than
polynomial space, is an open problem.
\end{quote}

Much progress on streaming clustering problems in Euclidean space has been due 
to coresets\cite{HPM04,Ch09,FS05,FFS06,FMSW10}.
But, similarly to the metric case,
methods for maintaining coresets on sliding windows for Euclidean 
clustering problems have been hard to come by.
Streaming insertion-only coreset techniques exist for the Euclidean
$k$-median and $k$-means problems in low-and high-dimensional spaces,
but to our knowledge no such results exist for the sliding window model.
We present the first such technique,
a general framework in which one can build coresets for
a broad class of clustering problems in Euclidean space.

\subsection{Our Contribution}
\paragraph{Metric clustering problems in the sliding window model.}
We present the first polylogarithmic space
$O(1)$-approximation for the metric $k$-median problem in the sliding window
model, answering the question posed by Babcock et al.~\cite{BDMO03}.
Our algorithm uses $O(k^3 \log^6 n)$-space and
requires update time $O(\poly(k, \log n) )$, with the exact
update time depending on which of the many existing
offline $O(1)$-approximations for $k$-median one chooses.
We also demonstrate how this result extends to a broad class of related
clustering problems, including $k$-means.
The one requirement of our result is a polynomial bounded on the
ratio of the optimal cost to the minimum inter-point distance
on any window of the stream. 

Braverman and Ostrovsky~\cite{BO07} introduced smooth histograms
as a method for adapting insertion-only streaming algorithms to the
sliding windows model for a certain class of functions,
which Braverman and Ostrovsky call \emph{smooth}.
Unfortunately, the $k$-median and $k$-means costs are not smooth functions,
so smooth histograms cannot be directly applied.
Our major technical contribution lies in the extension of
smooth histograms~\cite{BO07} to a class of clustering functions,
including $k$-median and $k$-means clustering,
that are less well-behaved than smooth functions.
We show that clustering problems $k$-median and $k$-means
do possess a property similar to smoothness
provided that a pair of conditions hold related to the cluster cardinalities
and costs of clustering solutions built on certain subsets of the stream.
We develop a streaming data structure that ensures that these two conditions
are satisfied where necessary
so that the core ideas behind the smooth histogram data structure
can be brought to bear.
Using the algorithms of~\cite{COP03} and~\cite{BO07},
we show that the bookkeeping necessary for our approach can be maintained
using in polylogarithmic space and time.

\paragraph{Euclidean clustering problems in the sliding window model.}
Merge-and-reduce is a generic method in computational geometry to implement
offline algorithms in the insertion-only streaming model.
Several well-known coreset constructions are conducive to this method,
including well-known coreset techniques for the $k$-median and $k$-means problems in both low-and high-dimensional
Euclidean spaces~\cite{HPM04,Ch09}.
We develop a sliding window algorithm that, given one of these insertion-only
streaming coresets of size $s$, maintains this coreset in the sliding window
model using $O(s^2\epsilon^{-2}\log n)$ space. 

To develop our generic framework, we consider a sequence $X$ of indices 
of arrival times of points as in the smooth histograms~\cite{BO06}. 
For each index $x_i\in X$ we maintain a coreset (using merge-and-reduce)
of points whose arrival times are between $x_i$ and current time $N$.
In particular, as the points in the interval $[x_i,N]$ arrive, 
we compute coresets for small subsets of points. 
Once the number of these coresets is big enough,
we \emph{merge} these coresets  and \emph{reduce} them by computing a
coreset on their union.
This yields a tree whose root is a coreset of its leaves,
which contain subsets of the points in the interval $[x_i,N]$. 
The well-known coreset techniques of~\cite{HPM04,Ch09,FFS06,FMSW10} 
mostly partition the space into small set of regions and from each region
take a small number of points either randomly or deterministically. 
Hence, at the root of the merge-and-reduce tree,
we have a partition from whose regions we take weighted coreset points. 

The crux of the smooth histograms~\cite{BO06} data structure
is to maintain a small set of indices $X$ 
such that for every two consecutive indices $x_i$ and $x_{i+1}$ all
intervals $[x_{i}+1,N],\dots,[x_{i+1}-1,N]$ have clustering cost which are
within $(1+\epsilon)$-fraction of each other.
By keeping only indices $x_i,x_{i+1}$,
we smooth the cost between these two indices. 
To this end, we look at the partition of the root of the merge-and-reduce tree
corresponding to arrival time $t\in[x_i,x_{i+1}]$. 
If there is a region in this partition with at most $\epsilon$-fraction of its
points in interval $[x_i,x_{i+1}]$, then we ignore the interval $[t,N]$;
otherwise we add index $t$ to $X$ and keep the coreset of points for the
interval $[t,N]$. 
We show using a novel application of VC-dimension and $\epsilon$-sample
theory~\cite{as-pm} that 
if we take small random samples from every region of the partitions in the
intermediate nodes of the merge-and-reduce tree, 
then the coreset points inside every region of the partition of the root of
this tree is a good approximation of 
the original points in that region.
Thus, testing whether $\epsilon$-fraction of coreset points of a region are
in the interval $[x_i,x_{i+1}]$ is a good approximation for testing whether
$\epsilon$-fraction of the original points of the region are in this interval. 
We also show that if for every region in the partition at most
$\epsilon$-fraction of its points are in the interval 
$[x_i,x_{i+1}]$, then ignoring the interval $[t,N]$ and keeping only
intervals $[x_i,N]$ and $[x_{i+1},N]$ 
loses at most $\epsilon$-fraction of the clustering cost of points in
interval $[t,N]$. 

Frahling and Sohler~\cite{FS05} developed a coreset technique
for $k$-median and $k$-means problems in low-dimensional spaces
in the dynamic geometric stream 
(i.e., a stream of insertions and deletions of points)
using the heavy hitters algorithm~\cite{cm05b} and a logarithmic sampling rate.
We observe that we can maintain their coreset in the sliding window model
using the heavy hitter algorithm and sampling 
techniques proposed for the sliding window model due to
Braverman, Ostrovsky and Zaniolo~\cite{BOZ12}.
However, their approach does not work for other well-known coreset techniques
for Euclidean spaces~\cite{HPM04,Ch09,FFS06,FMSW10},
motivating the need for a different technique,
which we develop in this paper.

\subsection{Related Work}
\label{sec:rel:wo}
Guha, Mishra, Motwani and O'Callaghan~\cite{GMMO00} presented the
first insertion-only streaming algorithm for the $k$-median problem.
They gave a $2^{O(1/\epsilon)}$-approximation
using $O(n^{\epsilon})$ space, where $\epsilon<1$.
Charikar, O'Callaghan, and Panigrahy~\cite{COP03},
subsequently developed an $O(1)$-approximation
insertion-only streaming algorithm using $O(k \log^2 n)$ space.
Their approach operates in phases, similarly to~\cite{GMMO00},
maintaining a set of $O(k \log n)$ candidate centers
that are reduced to exactly $k$ centers using an offline $k$-median
algorithm after the entire stream has been observed.

Slightly stronger results hold when the elements of the stream are points
in $d$-dimensional Euclidean space $\REAL^d$.
Har-Peled and Mazumdar~\cite{HPM04} developed a $(1+\epsilon)$-approximation
for $k$-median and $k$-means in the insertion-only
streaming model using (strong) coresets.
Informally, a strong $(k,\epsilon)$-coreset for $k$-median
is a weighted subset $S$ from some larger set of points $P$
that enables us to (approximately) evaluate the quality of a
candidate solution on $P$ using small space.
The coresets presented by Har-Peled and Mazumdar~\cite{HPM04} required
$ O( k \epsilon^{-d} \log{n} )$ space, yielding a streaming solution
using $O(k \epsilon^{-d} \log^{2d+2}{n})$ space via the famous
merge-and-reduce approach~\cite{bs-dspsd-80,AHV04}.
Har-Peled and Kushal~\cite{HPK05} later developed coresets of size $ O( k^2 \epsilon^{-d} )$ for $k$-median and $k$-means problems.
Unfortunately, these new coresets do not result in significant space
improvements in the streaming model.
Feldman, Fiat and Sharir~\cite{FFS06} later extended this type of
coreset to the case where centers can be lines or flats.

In high-dimensional spaces, Chen~\cite{Ch09} presented a technique
for building $(k,\epsilon)$-coresets of size $O(k^2d\epsilon^{-2}\log^2{n})$,
yielding via merge-and-reduce a streaming
algorithm requiring $O(k^2d\epsilon^{-2}\log^8 n)$ space.
Chen~\cite{Ch09} also presented a technique for general metric spaces,
which, with probability of success $1-\delta$, produces a coreset of size
$O(k\epsilon^{-2}\log n(k\log n+\log(1/\delta)))$.

To the best of our knowledge, there do not exist any results
to date for the $k$-median problem on arbitrary metric spaces
(often simply called the \emph{metric $k$-median} problem)
in the insertion-deletion streaming model.
In the geometric case, introduced by Indyk~\cite{Ind04}
under the name \emph{dynamic geometric data streams},
Frahling and Sohler~\cite{FS05}
have shown (via a technique distinct from that in~\cite{HPM04,HPK05})
that one can build a $(k,\epsilon)$-coreset for $k$-median or $k$-means
using $O(k^2 \epsilon^{-2d-4} \log^7 n)$ or $O(k \epsilon^{-d-2} \log n)$
space, respectively.

In comparison to the insertion-only and dynamic geometric streaming
models, little is known about the metric $k$-median problem in the sliding
window model, where the goal is to maintain a solution on the most
recent $W$ elements of the data stream.
To our knowledge, the only existing solution under this model
is the $O(2^{O(\frac{1}{\tau})})$-approximation given in~\cite{BDMO03},
where $\tau \in (0,1/2)$ is a user-specified parameter.
The solution presented therein requires $O(\frac{k}{\tau^4}W^{2\tau}\log^2 W)$
space and yields an initial solution using $2k$ centers,
which is then pared down to $k$ centers with no loss to the
approximation factor.


\paragraph{Outline.}
The remainder of the paper is organized as follows:
Section~\ref{sec:prelims} establishes notation and definitions.
Our main results are presented in Section~\ref{sec:metric},
which gives an algorithm for the $k$-median problem on sliding windows,
and Section~\ref{sec:euclid}, which presents an algorithm for maintaining
coresets for Euclidean clustering problems on sliding windows.
Additional results and proof details are included in the Appendix.


\section{Preliminaries}
\label{sec:prelims}
We begin by defining the clustering problems of interest
and establishing notation.

\subsection{Metric and Geometric $k$-Median Problems}
\label{prelim:def:k:median}
Let $(X,\dist)$ be a metric space where $X$ is a set of points and
$\dist:X\times X\rightarrow \REAL$
is a distance function defined over the points of $X$.
For a set $Q \subseteq X$,
we let $\dist(p,Q)=\min_{q \in Q} \dist(p,q)$ denote
the distance between a point $p\in X$ and set $Q$
and we denote by $\rho(Q)=\min_{p,q \in Q, p \neq q} \dist(p,q)$
the minimum distance between distinct points in $Q$.
We define $[a]=\{1,2,3,\cdots a\}$ and
$[a,b]=\{a,a+1,a+2,\cdots b\}$ for natural numbers $a\le b$.
When there is no danger of confusion,
we denote the set of points $\{p_a,p_{a+1},\dots,p_b\} \subset X$
by simply $[a,b]$.
For example, for function $f$ defined on sets of points,
we denote $f(\{p_a,p_{a+1},\dots,p_b\})$ by simply $f([a,b])$.

\begin{definition}[Metric $k$-median]
Let $P$ be a set of $n$ points in metric space $(X,d)$
and let $C=\{c_1,\dots,c_k\} \subseteq X$ be a set of $k$ points
called \emph{centers}.
A \emph{clustering} of point set $P$ using $C$ is a partition of $P$ such
that a point $p\in P$ is in partition $P_i$ if $c_i \in C$
is the nearest center in $C$ to $p$, with ties broken arbitrarily.
We call each $P_i$ a \emph{cluster}.
The $k$-median cost using centers $C$ is $\cost(P,C)=\sum_{p\in P} \dist(p,C).$
The metric $k$-median problem is to find a set $C^* \subset P$ of $k$ centers
satisfying
$ \cost(P,C^*) = \min_{C \subset P: |C|=k} \cost(P,C)$.
We let $\opt(P,k)=\min_{C \subset P: |C|=k} \cost(P,C)$ denote this
optimal $k$-median cost for $P$.
\end{definition}

\begin{definition}[(Euclidean) $k$-median Clustering]
Let $P$ be a set of $n$ points in a $d$-dimensional Euclidean Space $\REAL^d$ 
and $k$ be a natural number.
In the $k$-median problem,
the goal is to find a set $C=\{c_1,\cdots,c_k\} \subset \REAL^d$ of $k$ centers, 
that minimizes the cost
$ \cost(P,C) =\sum_{p\in P} \dist(p,C)$, 
where $\dist(p,C)=\min_{c_i\in C} \dist(p,c_i)$ is the Euclidean distance between $p$ 
and $c_i$. 
\end{definition}

\begin{definition}[$(k,\epsilon)$-Coreset for $k$-Median Clustering]
Let $P$ be a set of $n$ points in $d$-dimensional Euclidean Space $\REAL^d$ 
and let $k$ be a natural number. A set $S\subseteq \REAL^d$ is a 
$(k,\epsilon)$-coreset for $k$-median  clustering
if for every set $C=\{c_1,\cdots,c_k\} \subset \REAL^d$ of $k$ centers we have 
$
  |\cost(P,C)-\cost(S,C)|\le \epsilon \cdot\cost(P,C)
$.
\end{definition}

\subsection{The Sliding Window Model}
\label{prelim:sliding:window}
Let $(X,d)$ be a metric space and
$P \subseteq X$ a point set of size $|P|=n$.
In the insertion-only streaming model~\cite{AMS99,HPM04,COP03},
we think of a (possibly adversarial) permutation $p_1,p_2,\cdots,p_n$
of $P$, presented as a data stream.
We assume that we have some function $f$, defined on sets of points.
The goal is then to compute the (approximate) value of $f$
evaluated on the stream, using $o(n)$ space.
We say that point $p_N$ \emph{arrives} at time $N$.

The \emph{sliding window model}~\cite{DGIM02} is a generalization
of the insertion-only streaming model in which
we seek to compute function $f$ over only the most recent elements
of the stream.
Given a current time $N$, we consider a window $\mathcal{W}$ of size $W$
consisting of points $p_s,p_{s+1},\dots,p_N$, where $s=\max\{1,N-W+1\}$.
We assume that $W$ is such that we cannot store all of
window $\mathcal{W}$ in memory.
A point $p_i$ in the current window $\mathcal{W}$ is called \emph{active}.
At time $N$, point $p_i$ for which $i<N-W+1$
is called \emph{expired}.

\subsection{Smooth Functions and Smooth Histograms}
\label{prelim:smooth:function}
\begin{definition}[$(\epsilon,\epsilon')$-smooth function~\cite{BO06}]
Let $f$ be a function defined on sets of points,
and let $\epsilon,\epsilon' \in (0,1)$.
We say $f$ is an \emph{$(\epsilon,\epsilon')$-smooth function} if
$f$ is non-negative (i.e., $f(A)\ge 0$ for all sets $A$),
non-decreasing (i.e., for $A \subseteq B$, $f(A)\le f(B)$),
and polynomially bounded (i.e., there exists constant $c > 0$ such that
$f(A) = O(|A|^c)$ ) and for all sets $A,B,C$
\[
  f(B) \ge (1-\epsilon) f(A\cup B)
  \enspace \text{ implies } \enspace
  f(B\cup C)\ge (1-\epsilon') f(A\cup B \cup C) .
\]
\end{definition}
Interestingly, a broad class of functions fit this definition.
For instance, sum, count, minimum, diameter, $L_p$-norms, frequency moments
and the length of the longest subsequence are all smooth functions.


Braverman and Ostrovsky~\cite{BO06}
proposed a data structure called \emph{smooth histograms}
to maintain smooth functions on sliding windows.

\begin{definition}[Smooth histogram~\cite{BO06}]
Let $0 < \epsilon < 1, 0 < \epsilon' < 1$ and $\alpha>0$,
and let $f$ be an $(\epsilon,\epsilon')$-smooth function.
Suppose that there exists an insertion-only streaming algorithm
$\mathcal{A}$ that computes an $\alpha$-approximation $f'$ of $f$.
The \emph{smooth histogram} consists of an increasing set of indices
$X_N=\{ x_1,x_2,\cdots,x_t=N \}$ and $t$ instances
$\mathcal{A}_1,\mathcal{A}_2,\cdots,\mathcal{A}_t$ of $\mathcal{A}$
such that
\begin{enumerate}
\item Either $p_{x_1}$ is expired and $p_{x_2}$ is active or $x_1=0$.
\item For $1 < i< t-1$ one of the following holds
        \begin{enumerate}
        \item $x_{i+1}=x_i+1$ and $f'([x_{i+1},N])\le (1-\epsilon')f'([x_{i},N])$,
        \item $f'([x_{i+1},N])\ge (1-\epsilon)f'([x_{i},N])$ and if $i\in [t-2]$, $f'([x_{i+2},N])\le (1-\epsilon')f'([x_{i},N])$.
        \end{enumerate}
\item $\mathcal{A}_i=\mathcal{A}([x_i,N])$ maintains $f'([x_i,N])$.
\end{enumerate}
\end{definition}
Observe that the first two elements of sequence $X_N$ always sandwich
the current window $\mathcal{W}$, in the sense that $x_1\le N-W\le x_2$.
Braverman and Ostrovsky~\cite{BO06} used this observation
to show that at any time $N$, one of either $f'([x_1,N])$ or $f'([x_2,N])$
is a good approximation to $f'([N-W,N])$, and is thus a good approximation
to $f([N-W,N])$.
In particular, they proved the following theorem.
\begin{theorem}[\cite{BO06}]
\label{thm:BO06}
Let $0< \epsilon,\epsilon'<1$ and $\alpha,\beta>0$,
and let $f$ be an $(\epsilon,\epsilon')$-smooth function.
Suppose there exists an insertion-only streaming algorithm
$\mathcal{A}$ that calculates an $\alpha$-approximation
$f'$ of $f$ using $g(\alpha)$ space and $h(\alpha)$ update time.
Then there exists a sliding window algorithm that maintains a
$(1\pm (\alpha+\epsilon))$-approximation of $f$ using
$O(\beta^{-1}\cdot\log n \cdot(g(\alpha)+\log n))$ space
and $O(\beta^{-1}\cdot\log n\cdot h(\alpha))$ update time.
\end{theorem}

\paragraph{VC-Dimension and $\epsilon$-Sample.}
\label{sec:vc:dim}
We briefly review the definition of VC-dimension and $\epsilon$-sample
as presented in Alon and Spencer's book \cite{as-pm} in Appendix \ref{appendix:euclid}.

\section{Metric $k$-Median Clustering in Sliding Windows}
\label{sec:metric}

We introduce the first polylogarithmic-space $O(1)$-approximation for
metric $k$-median clustering in the sliding window model.
Our algorithm requires $O(k^3 \log^6 n)$ space and
$O( \poly(k,\log n) )$ update time.
We note that our algorithm is easily modified to accommodate $k$-means
clustering.  This modified algorithm will have the same time and space bounds
with a larger approximation ratio that is nevertheless still $O(1)$.

\subsection{Smoothness}

$k$-median and $k$-means are not smooth (see the Appendix for an example),
so the techniques of~\cite{BO06} do not apply directly,
but Lemma~\ref{bounding} shows that $k$-median clustering does
possess a property similar to smoothness.

\begin{definition}[$\tri$-approximate Triangle Inequality]
Non-negative symmetric function
$d~:~X~\times~X~\rightarrow~\mathbb{R}_{\ge 0}$
satisfies the \emph{$\tri$-approximate triangle inequality}
if $d(a,c)\le\tri \left( d(a,b)+d(b,c) \right)$
for every $a,b,c\in X$.
\end{definition}
We note that a metric $d$ satisfies the $1$-approximate
triangle inequality by definition
and that for any $p \ge 1$,
$d^p$ obeys the $2^{p-1}$-approximate triangle inequality,
since
$(x+y)^p \le 2^{p-1}(x^p+y^p)$
for all non-negative $x,y$ and $p \ge 1$.
Thus, if $p = O(1)$, Theorem~\ref{mainTheorem}
provides an $O(1)$-approximation
for the clustering objective $\sum d^p(x,C)$.
The case $p=2$ yields an $O(1)$-approximate solution for $k$-means.

\begin{definition}
Let $P$ and $C = \{c_1,\dots,c_k\}$
be sets of points from metric space $(\mathcal{X},d)$.
A map $t : P \rightarrow C$ is called a
\emph{clustering map} of $P$ for set $C$.
If $\sum_{x \in P} d(x,t(x)) \le \beta \cdot \opt(P,k)$,
then we say that $t$ is a \emph{$\beta$-approximate clustering map} of $P$
for set $C$.
The difference between a clustering map $t(x)$ and the intuitive map
$\arg \min_{c \in \{c_1, \dots, c_k\}} d(x,c)$ is that $t$ need not map each
point $x$ to its nearest center.
\end{definition}

\begin{lemma}\label{bounding}
Let $d$ be a non-negative symmetric function on $X \times X$
satisfying the $\tri$-approximate triangle inequality
and let $A,B \subset X$ be two disjoint sets of points
such that
\begin{enumerate}
\item $\opt(A\cup B,k) \le \gamma \opt(B,k)$. \label{cond1:cost}
\item There exists $\beta$-approximate clustering map $t$ of $A\cup B$
	such that
	$\forall i\in[k] :|t^{-1}(c_i)\cap A|\le |t^{-1}(c_i)\cap B|$.
	\label{cond2:card}
\end{enumerate}
\vspace{-5mm}
Then for any $C \subseteq X$ we have
$ \opt(A\cup B\cup C,k) \le (1+\tri+\beta \gamma\tri)\opt(B\cup C,k). $
\end{lemma}
\begin{proof}
Let $s$ be an optimal clustering map for $B\cup C$
(i.e., $s$ is $1$-approximate). Then
$$ \opt(A \cup B \cup C,k) \le \sum_{a \in A} d(a, s(a))
			+ \sum_{x \in B \cup C} d(x, s(x)). $$
The second term is $\opt(B\cup C,k)$,
and we can bound the first term
by connecting each element of $t^{-1}(c_i)\cap A$
to a unique element in $t^{-1}(c_i)\cap B$,
applying the $\tri$-approximate triangle inequality,
and using the fact that $|t^{-1}(c_i)\cap A|\le |t^{-1}(c_i)\cap B|$
for all $i \in [k]$.
Details are provided in the Appendix.
\end{proof} 

If we could ensure that the inequality conditions
required by Lemma~\ref{bounding} hold, then we could apply the ideas from
smooth histograms. The following two lemmas suggest a way to do this.

\begin{lemma} \label{blackbox}
Let $A \cup B$ be a set of $n$ points received in an insertion-only
stream, appearing in two phases
so that all points in $A$ arrive before all points in $B$,
and assume that the algorithm is notified when the last point from $A$ arrives.
Using $O(k \log^2 n)$ space, it is possible to compute an
$O(1)$-approximate clustering map $t$ for $A \cup B$
as well as the exact values of
$\{(|t^{-1}(c_i) \cap A|, |t^{-1}(c_i) \cap B|)\}_{i \in [k]}$.
\end{lemma}
\begin{proof}
Given a set of points $P$ presented in a stream $D$,
the \textPLS algorithm presented in~\cite{COP03}
uses $O(k \log^2 n)$ space to compute a weighted set $S$
such that $\cost(D,S) \le \alpha \opt(D,k)$ for some constant $\alpha$
(\cite{COP03}, Theorem 1).
Using a theorem from~\cite{GMMO00}, it is shown in~\cite{COP03} that
$\opt(S,k) \le 2(1+\alpha)\opt(D,k)$.
It follows immediately that
running an offline $\xi$-approximation for $k$-median on $S$
yields a set of $k$ centers that constitutes an
$(\alpha + 2\xi(1+\alpha))$-approximate $k$-median solution
for the original stream $D$.

The \textPLS algorithm uses as a subroutine
the online facility location algorithm due to Meyerson~\cite{M01}.
Thus, each point in $S$ can be viewed as a facility serving
one or more points in $P$.
Therefore, running the \textPLS algorithm on stream $D$
yields a map $r: P \rightarrow S$ such that $r(p) \in S$
is the facility that serves point $p \in P$.
Running a $\xi$-approximation on the set $S$
to obtain centers $\{ c_1, \dots, c_k\}$,
yields a map $q: S \rightarrow \{ c_1, \dots, c_k\}$
such that point $s \in S$ is connected to $q(s)$.

Given disjoint multisets $A$ and $B$, \textPLS
yields maps $r_{A} : A \rightarrow S_A$ and $r_B : B \rightarrow S_B$.
Running the $\xi$-approximation on $S_A \cup S_B$,
we obtain a map
$q : S_{A} \cup S_{B} \rightarrow \{ c_1, \dots, c_k\}$,
from which we have a $\beta$-approximate clustering map of $A \cup B$
given by $t(x) = q(r_A(x))$ if $ x \in A$
and $t(x) = q(r_B(x))$ if $ x \in B$,
from which we can directly compute
$|t^{-1}(c_i) \cap A| = |q^{-1}(c_i) \cap S_A|$
and similarly for $|t^{-1}(c_i) \cap B|$.
\end{proof}

The previous lemma showed how we can check whether condition~\ref{cond2:card}
in Lemma~\ref{bounding} holds over two phases.
We now extend this to the case where the stream has an arbitrary
number of phases.

\begin{lemma} \label{maintainSizes}
Let $A_1 \cup \dots \cup A_Z$ be a set of $n$ points in an
insertion-only stream, arriving in phases $A_1,A_2,\dots,A_Z$,
and assume that the algorithm is notified at the end of each phase $A_i$.
Using $O(Z^2 k \log^2 n)$ space, one can compute
for every $1 \le j < \ell \le Z$ a $\beta$-approximate clustering
map $t_{j,\ell}$ for $A_1 \cup \dots \cup A_Z$ and the exact values of
$\{|t_{j,\ell}^{-1}(c_i) \cap (A_{j} \cup \dots \cup A_{\ell-1})|,
|t_{j,\ell}^{-1}(c_i) \cap (A_{j} \cup \dots \cup A_Z)|\}_{i \in [k]}$
for that map.
\end{lemma}
\begin{proof}
This lemma is a natural extension of Lemma~\ref{blackbox}.
Details of the proof are in the Appendix.
\end{proof}

Lemma~\ref{maintainSizes} suggests one way to ensure that the conditions
of Lemma~\ref{bounding} are met-- simply treat every point
as a phase-- but this would require running
$O(W^2)$ instances of \textPLSns, which would be infeasible.
We would like to ensure that the
conditions~\ref{cond1:cost} and~\ref{cond2:card} in~\ref{bounding} hold,
while running at most $T \ll W$ many instances of \textPLSns.
We can achieve this by starting a new phase only when one of
these conditions would otherwise be violated.
We will show in Lemma~\ref{spaceBound} that this strategy incurs only
polylogarithmically many phases.


\subsection{Algorithm for Sliding Windows}
\label{SW-section}
Algorithm~\ref{alg:kmedian:approx:SW} produces an approximate $k$-median
solution on sliding windows.
The remainder of the section will establish properties of this algorithm,
culminating in the main result given in Theorem~\ref{mainTheorem}.

The bulk of the bookkeeping required by Algorithm~\ref{alg:kmedian:approx:SW}
is performed by the {\sc Update} subroutine, defined in
Algorithm~\ref{alg:kmedian:update}.
In the spirit of~\cite{BO07},
central to our approach is a set of indices $\{X_1,X_2,\dots,X_T\}$.
Each $X_i$ is the arrival time of a certain point from the stream.
Algorithm~\ref{alg:kmedian:approx:SW} runs $O(T)$ instances of \textPLS
on the stream, with the $i$-th instance running starting
with point $p_{X_i}$. Denote this instance by $\mathcal{A}(X_i)$.
The \textPLS algorithm on input $P = \{p_i,p_{i+1},\dots,p_j\}$
constructs a weighted set $S$ and a score $R$ such that
$\opt(P,k) \le \cost(P,S) \le R \le \alpha \opt(P,k)$.
To check the smoothness conditions in Lemma~\ref{bounding}
we will use the solutions built up by certain instances of \textPLSns.
We keep an array of $O(T^2)$ buckets, indexed as
$B(X_i,X_j)$ for $1 \le i < j \le T$.
In each bucket we store $B(X_i,X_j) = (S_{ij},R_{ij})$ where
$S_{ij} = S(X_i,X_j)$ and $R_{ij} = R(X_i,X_j)$ are, respectively,
the weighted set and the cost estimate produced by an instance
of \textPLS running on the substream $\{p_{X_i},\dots,p_{X_j} \}$.

Concretely, we run instance $\mathcal{A}(X_i)$ on the stream
starting with point $p_{X_i}$. At certain times, say time $N$,
we will need to store the solution currently built up by this instance.
By this we mean that we copy the weighted set and
cost estimate as constructed by $\mathcal{A}(X_i)$ on points
$\{p_{X_i},\dots,p_N\}$ and store them in bucket $B(X_i,N)$.
We denote this by $B(X_i,N) \gets \store(\mathcal{A}(X_i))$.
Instance $\mathcal{A}(X_i)$ continues running after this operation.
We can view each $B(X_i,X_j)$ as a snapshot of the \textPLS
algorithm as run on points $\{p_{X_i},\dots,p_{X_j}\}$.
As necessary, we terminate \textPLS instances and initialize new ones over the
course of the algorithm.
We assume an offline $k$-median $O(1)$-approximation
algorithm $\mathcal{M}$, and denote by $\mathcal{M}(P)$
the centers returned by running this algorithm on point set $P$.

\begin{algorithm*}
\textbf{Input:} A stream of points $D = \{p_1,p_2,\dots,p_n\}$ from
metric space $(\mathcal{X},\dist)$, window size $W$

\noindent
\textbf{Update Process, upon the arrival of new point $p_N$:}

\begin{algorithmic}[1]
	\FOR{$i = 1,2,\dots,T$}
		\STATE $B(X_i,N) \gets \store( \mathcal{A}(X_i) )$
	\ENDFOR
	\STATE Begin running $\mathcal{A}(N)$
	\STATE \update()
\end{algorithmic}

\noindent
\textbf{Output:}
Return the centers and cost estimate from bucket $B(X_1, N)$

\caption{Metric $k$-median in Sliding Windows}
\label{alg:kmedian:approx:SW}
\end{algorithm*}


\begin{algorithm*}

\begin{algorithmic}[1]
	\STATE If $X_2 > N-W$, then $i \gets 1$. Otherwise, $i \gets 2$.
	\WHILE{$i \le T$} \label{dropOPT}
		\STATE $j \gets$ the maximal $j'$ such that
			$\beta R(X_i,N) \le \gamma R(X_{j'},N)$.
			If none exist, $j \gets i+1$ \label{assignJ}
		\STATE $C \gets \mathcal{M}( S(X_i,X_j) )$
		\WHILE{$i < j$} \label{dropsize}
			\STATE $\KHmark(X_i)$
			\STATE $i \gets$ the maximal $\ell$ such that 
	$|t^{-1}(c)\cap \mathcal{S}_{i,\ell}| \le |t^{-1}(c)\cap \mathcal{S}_{\ell,T}|$ for all $c \in C$. If none exist, $i \gets i+1$ \label{getMax}
		\ENDWHILE
	\STATE $\KHmark(X_j)$
	\STATE $i \gets j + 1$
	\ENDWHILE \label{endDropOPT}
	\STATE For all unmarked $X_i$, terminate instance $\mathcal{A}(X_i)$
	\STATE Delete all buckets $B(X_i,X_j)$ for which either $X_i$ or $X_j$
		is unmarked.
	\STATE Delete all umarked indices $X_i$;
		relabel and unmark all remaining indices.
\end{algorithmic}
\caption{ {\sc Update} : prevents Algorithm~\ref{alg:kmedian:approx:SW}
	from maintaining too many buckets. }
\label{alg:kmedian:update}
\end{algorithm*}

\begin{lemma} \label{approxCost}
For any index $m \le N$, let $s$ be the maximal index such that
$[m,N] \subseteq [X_s, N]$.
Then $\opt([m,N]) \le (2+\beta \gamma) \opt([X_s, N])$.
\end{lemma}
\begin{proof}
If $X_s = m$, there is nothing to prove, so assume $X_s < m$.
This implies that index $m$ was deleted at some previous time $Q$.
The result follows by considering the state of the algorithm
at this time.
Algorithm~\ref{alg:kmedian:approx:SW} maintains the conditions required
by Lemma~\ref{bounding} to ensure that
for any suffix $C$,
$\opt([X_i,Q]\cup C) \le (2+\beta \gamma) \opt([m,Q]\cup C)$.
Letting $C = [Q+1, N]$ yields the result.
Details are given in the Appendix.
\end{proof}

In what follows, let $\opt' = \opt(W,k) / \rho(W)$ and
$n = |W|$, the size of the window.

\begin{lemma} \label{spaceBound}
Algorithm \ref{alg:kmedian:approx:SW} maintains
$O(k \log n \log \opt')$ buckets.
\end{lemma}
\begin{proof}
Each iteration of the loop on Line~\ref{dropOPT}
decreases $R(X_i)$ by a factor of $\gamma / \beta$,
so this loop is executed $O(\log_{\gamma / \beta} \opt')$ times.
In each iteration of the loop on Line~\ref{dropsize},
the size of at least one of the $k$ clusters decreases by half.
Each set has size at most $n$, so
this loop is executed $O(k \log_2 n)$ times.
Each execution of each loop stores one bucket,
so in total $O(k \log n \log \opt')$ buckets are stored
in these nested loops.
\end{proof}

\begin{lemma} \label{updateTime}
Assuming $\opt' = \poly(n)$, Algorithm~\ref{alg:kmedian:approx:SW}
requires $O( \poly(k, \log n) )$ update time.
\end{lemma}
\begin{proof}
The runtime of Algorithm~\ref{alg:kmedian:approx:SW} is dominated by the
$O(T k^2 \log^2 n)$ time required to partition the buckets
and that $T = O(k \log^2 n)$ by Lemma~\ref{spaceBound}.
A more detailed proof is given in the Appendix.
\end{proof}

\begin{theorem} \label{mainTheorem}
Assuming $\opt' = \poly(n)$,
there exists an $O(1)$-approximation for the metric
$k$-median problem on sliding windows using $O(k^3 \log^6 n)$ space and
$O( \poly(k, \log n))$ update time.
\end{theorem}
\begin{proof}
Using Algorithm~\ref{alg:kmedian:approx:SW},
we output a $\beta$-approximation for $[X_1, N]$,
which includes the current window $\mathcal{W}$.
By Lemma~\ref{approxCost},
$\opt([X_1,N],k) \le (2+\beta \gamma) \opt(\mathcal{W},k)$, and thus
$R(X_1,N) \le \beta (1 + \lambda +\beta \gamma \lambda) \opt(\mathcal{W},k)$.
Let $C$ be the approximate centers for $[X_1,N]$.
We have the following inequalities:
\begin{align*}
	\cost([X_1,N],C) &\le \beta \opt([X_1,N]) \\
	\opt([X_1,N])    &\le (1 + \lambda +\beta \gamma \lambda)
			\opt(\mathcal{W},k) \\
	\cost(W,C)   &\le \cost([X_1,N],C),
\end{align*}
where the last equation follows from the fact that $[X_1,N]$
contains the current window.
Connecting these inequalities, we have
$\cost(\mathcal{W},C)
\le \beta(1 + \lambda +\beta \gamma \lambda) \opt(\mathcal{W},k)$, as desired.

For the space bound,
note that for each $1 \le i < j \le T$,
bucket $B(X_i,X_j)$ contains the output of an instance of $\textPLS$.
Each of these $O(T^2)$ instances requires $O(k \log^2 n)$ space,
and $T = O(k \log n \log \opt')$ by Lemma~\ref{spaceBound},
so our assumption that $\opt' = \poly(n)$ implies that we use
$O(k^3 \log^6 n)$ space in total.
\end{proof}

\section{Euclidean Coresets on Sliding Windows}
\label{sec:euclid}
In this section we first explain a coreset technique that unifies many of the known
coreset techniques. 
Then we explain the merge-and-reduce method.
Finally we develop our sliding window algorithm 
for coresets. 

\paragraph{A Unified Coreset Technique Algorithm.}
\label{sec:unify:coreset}
Many coreset technique algorithms for Euclidean spaces partition
the point set into small \emph{regions} (sometimes called \emph{cells}) 
and take a small number of points from each region of the partition either
randomly or deterministically. 
For each region, each of the chosen points is assigned a weight,
which is often the number of points in that region 
divided by the number of chosen points from that region. 
Some of the well-known coreset techniques that are in
this class are 
(1) the coreset technique  due to Har-Peled and Mazumdar
for the $k$-median and the $k$-means in low dimensional spaces~\cite{HPM04};
(2) the coreset technique due to Chen for the $k$-median
and the $k$-means problems in high-dimensional spaces~\cite{Ch09};
(3) the coreset technique due to
Feldman, Fiat and Sharir for the $j$-subspace problem in low dimensional
spaces~\cite{FFS06};
(4) the coreset technique due to Feldman, Monemizadeh,
Sohler and Woodruff for the $j$-subspace problem in high-dimensional
spaces~\cite{FMSW10}.
We unify this class of coreset techniques in
Algorithm \ref{alg:unify:coreset}.
In the sequel, we will use this unified view to
develop a sliding window streaming algorithm for this
class of coreset techniques. We will give the proofs for $k$-median and the $k$-means in low-and high-dimensional spaces 
and we defer the proofs for the $j$-subspace problem in low-and high-dimensional
spaces to the full version of this paper.

\begin{algorithm*}
\textbf{Input:} A set $P$ of $n$ points, a constant $c$ and two parameters
	$0< \epsilon,\delta\le 1$. 

\textbf{Algorithm:}
\begin{algorithmic}[1]
    \STATE Suppose we have a $(k,\epsilon)$-coreset technique $\mathcal{CC}$ that returns a partition $\Lambda_{\mathcal{K}}$ of $\REAL^d$.      
    \STATE Let $d_{VC}$ be the VC-D. of range space $(P,\mathcal{R})$ s.t. $\mathcal{R}$ is the set of shapes in $\REAL^d$ similar to regions in $\Lambda_{\mathcal{K}}$.
    \STATE Suppose  $\mathcal{CC}$ samples a set of size $s_{\mathcal{CC}}=f(n,d,\epsilon,\delta)$ from $R$ where $s_{\mathcal{CC}}$ is a function of $n,d,\epsilon,\delta$. 
    \STATE For each region $R\in \Lambda_{\mathcal{K}}$ we treat a weighted point $p$ of weight $w_p$ as $w_p$ points at coordinates of $p$.
    \STATE Sample $r=\min\big( |R|, \max\big(s_{\mathcal{CC}},O(d_{VC}\epsilon^{-2}\log(n)\cdot \log(\frac{d_{VC}\log(n)}{\epsilon\delta}))\big)\big)$ points uniformly at random.
    \STATE Each such a sampled point receives a weight of $n_R/r$ where $n_R$ is the number of points in region $R$.\
    \STATE Let $\mathcal{K}$ be the union of all (weighted) sampled points that are taken from regions in partition $\Lambda_{\mathcal{K}}$. 
\end{algorithmic}
  
\textbf{Output:} A coreset $\mathcal{K}$ of $P$ and its partition $\Lambda_{\mathcal{K}}$.

\caption{ Algorithm $\mathcal{A}$: A unified coreset technique}
\label{alg:unify:coreset}
\end{algorithm*}


\begin{lemma}
\label{lem:coreset:eps:slack} 
Let $P$ be a point set of size $n$ in a $d$-dimensional Euclidean space
$\REAL^d$ and $0< \epsilon \le 1$ be a parameter. 
Suppose we invoke one of the $(k,\epsilon)$-coreset techniques
of~\cite{HPM04} or ~\cite{Ch09}
and let $\mathcal{K}$ be the reported coreset and
$\Lambda_{\mathcal{K}}$ be the corresponding partition of $P$.
Suppose that for every region $R\in \Lambda_{\mathcal{K}}$ containing
$n_R$ points from $P$, we delete or insert up to $\epsilon n_R$ points.
Let $\mathcal{K}'$ be the coreset reported by Algorithm \ref{alg:unify:coreset}
after these deletions or insertions. 
Then, $\mathcal{K}'$ is a $(k,\epsilon)$-coreset of $\mathcal{K}$. 
\end{lemma}
\begin{proof}
Proofs for~\cite{HPM04} and ~\cite{Ch09} are in
Sections~\ref{sec:HM:coreset} and ~\ref{sec:chen:coreset},
respectively.
\end{proof}

\paragraph{Merge and Reduce Operation.}
The merge and reduce method inspired by a complete binary tree is a generic 
method in computational geometry to implement non-streaming algorithms 
in the insertion-only streaming model. 
Let $P$ be a set of $n$ points
\footnote{ Here we assume $n$ is known in advance. 
The case where $n$ is not known in advance can be accommodated using repeated
guesses for $n$. See, for example~\cite{HPM04}.},
presented in a streaming fashion.
The pseudocode of merge and reduce operation is given below. 
In this pseudocode, we use buckets
$\{B_1,B'_1, B_2,B'_2, \cdots, B_i, B'_i,\cdots,
B_{\log (n)-1},B'_{\log (n)-1},B_{\log(n)}\} ,$ 
where buckets $B_i$ and $B'_i$ can be considered as buckets in level $i$ of the merge-and-reduce tree 
and are of size $x_i$, which will be determined for each concrete problem. 
All buckets $B_i,B'_i$ for $i \in [\log n]$ are initialized to zero in the beginning of the stream.

\begin{algorithm*}
\textbf{Input:} A stream $S=[p_r,p_{r+1},p_{i+2},\cdots,p_{N-1}]$ of length $|S|=n^{c}$ for a constant $c$ and a point $p_N$.

\noindent
\textbf{Update Process, upon the arrival of new point $p_N$:}
\begin{algorithmic}[1]
    \STATE Let $i=1$ and add $p_N$ to $B_i$ if $B_i$ is not full, otherwise to $B'_i$.
    \WHILE{$B_i$ and $B'_i$ are both full and $i\le \log(n)$}
          \STATE Compute coreset $Z$ and partition $\Lambda_Z$ using Algorithm $\mathcal{A}(B_i\cup B'_i,\epsilon_i=\epsilon/(2\log n),\delta/n^c)$. 
          \STATE Delete the points of buckets $B_i$ and $B'_i$. 
          \STATE Let $B_{i+1}$ be $Z$ if $B_{i+1}$ is empty, otherwise let $B'_{i+1}$ be $Z$. 
          \STATE Let $i=i+1$. 
       \ENDWHILE
\end{algorithmic}

\textbf{Output:}  Return coreset $S_X=\cup_{i=1}^{\log(n)} B_i \cup B'_1$ and partition $\Lambda_{S_X}=\cup_{i=1}^{\log(n)} \Lambda_{B_i} \cup \Lambda_{B'_1}$.
\caption{{\sc MergeReduce} Operation}
\label{alg:merge:reduce}
\end{algorithm*}

The next lemma shows that the $(k,\epsilon)$-coreset
maintained by Algorithm {\sc MergeReduce} 
well-approximates the density of subsets of point set $P$ within every
region of partition $\Lambda_{B_i}$ for $i\in[\log n]$. 
The proof of this lemma is given in
Appendix~\ref{app:proof:bucket:sum:small:error}. 

\begin{lemma}
\label{lem:bucket:sum:small:error}
Let $B_i$ be the bucket at level $i$ of Algorithm {\sc MergeReduce} with $(k,\epsilon)$-coreset $B_i$ and partition $\Lambda_{B_i}$. 
Suppose the original points in the subtree $B_i$ is subset $P_i \subseteq P$.  
For every region $R_i\in \Lambda_{B_i}$,  
\[
   \begin{split}
     \big| |P_i \cap R_i| - |B_i \cap R_i| \big|
      \le \sum_{\text{level } j=2}^{i-1} \sum_{\text{node } x_j \text{ in level } j } \epsilon_{j} \big(\sum_{R\in \Lambda_{x_j}}  (|O_{x_j} \cap R|)  \big)  
      \enspace ,
    \end{split}
\]
where $O_{x_j}$ is a multi-set of points at node $x_j$ in level $i$ of the merge-and-reduce tree
such that for every point $p\in O_{x_j}$ with weight $w_p$,
we add $w_p$ copies of $p$ to $O_{x_j}$.  
\end{lemma}

Next we show the error of Lemma~\ref{lem:bucket:sum:small:error} is small.

\begin{lemma}
\label{lem:coreset:number:same:original}
Let $B_i$ be the bucket at level $i$ of Algorithm {\sc MergeReduce} with $(k,\epsilon)$-coreset $B_i$ and partition $\Lambda_{B_i}$. 
Suppose the original points in the subtree $B_i$ is subset $P_i \subseteq P$.  
For every $j\in[\log n]$, if we replace $\epsilon_j$ by $\epsilon/(2j)$,
the error
$ \sum_{\text{level } j=2}^{i-1} \sum_{\text{node } x_j \text{ in level } j }
  \epsilon_{j} \big(\sum_{R\in \Lambda_{x_j}}  (|O_{x_j} \cap R|)  \big)  $
of $\big| |P_i \cap R_i| - |B_i \cap R_i| \big|$  is $\epsilon$-fraction of the cost of $P_i$ 
in terms of $k$ arbitrary $j$-dimensional subspaces and so can be ignored.
\end{lemma}

\begin{proof} 
Let us look at the sub-terms in the error term
$ \sum_{\text{level } j=2}^{i-1}
  \sum_{\text{node } x_j \text{ in level } j }
  \epsilon_{j} \big(\sum_{R\in \Lambda_{x_j}}  (|O_{x_j} \cap R|)  \big)$.
For fixed node $x_j$ in level $j$,
$\epsilon_{j} \big(\sum_{R\in \Lambda_{x_j}}  (|O_{x_j} \cap R|)  \big)$
is the $\epsilon_j$-fraction change in region $R\in \Lambda_{x_j}$
of Algorithm  \ref{alg:unify:coreset}. 
Using Lemma~\ref{lem:coreset:eps:slack},
for each one of the coreset techniques in~\cite{HPM04} and ~\cite{Ch09}, 
the new $(k,\epsilon)$-coreset after these changes in every region is
again a $(k,2\epsilon)$-coreset of point set $P_i$. 
Here we use this fact that a $(k,\epsilon)$-coreset of a
$(k,\epsilon)$-coreset of $P$ is a $(k,2\epsilon)$-coreset of $P$. 
We have $i$ levels, each one of which is a $(k,\epsilon)$-coreset of $P_i$. 
Thus, the error term is $i$ times the error of one of one
$(k,\epsilon)$-coreset of $P_i$. 
If we replace $\epsilon_j$ by $\epsilon/(2j)$,
the overall error would be the same the error of one $(k,\epsilon)$-coreset
of $P_i$, which can be ignored. In Algorithm {\sc MergeReduce}
we replace $\epsilon_i=\frac{\epsilon}{2\log n}$ for all levels $i\in[\log n]$. 
\end{proof}

\paragraph{Coreset Maintenance in Sliding Windows.}
In this section we develop Algorithm {\sc SWCoreset},
a sliding window streaming algorithm for the class of coreset techniques 
(including the coreset techniques
of~\cite{HPM04},~\cite{Ch09},~\cite{FFS06}, and~\cite{FMSW10})
that fit into Algorithm \ref{alg:unify:coreset}. 
We show that the number of $(k,\epsilon)$-coresets that we maintain  
is upper bounded by the size of one $(k,\epsilon)$-coreset times $O(\epsilon^{-2}\log n)$.

\begin{algorithm*}
\textbf{Input:} A stream $S=[p_1,p_2,p_3,\dots,p_N,\dots,p_n]$ of points $\REAL^d$.

\textbf{Output:} A coreset $\mathcal{K}_{x_1}$ for window $W$, i.e., points $\{p_{N-W+1},\dots, p_N\}$.

\noindent
\textbf{Update Process, upon the arrival of new point $p_N$:}

\begin{algorithmic}[1]
	\FOR{$x_i\in X=[x_1,x_2,\dots,x_t]$  where $x_i\in \{1,\dots,N\}$}    \label{beginning}
                \STATE Let $(\mathcal{K}_{x_i},\Lambda_{\mathcal{K}_{x_i}})$={\sc MergeReduce}$([x_i,N]=\{p_{x_i},\dots, p_{N-1}\},p_N)$ be the coreset and its partition.
	\ENDFOR
         \STATE \label{increment} Let $t=t+1$,  $x_{t}=N$.
         \STATE Let $(\mathcal{K}_{x_t},\Lambda_{\mathcal{K}_{x_t}})$={\sc MergeReduce}$(\{\},p_N)$ be the coreset and the partition of single point $p_N$.
	\FOR{ $i=1$ to $t-2$}   \label{cleanup}
               \STATE Find greatest $j> i$  s.t. there is a region $R$ in partition $\Lambda_{\mathcal{K}_{x_j}}$ whose at most $\epsilon w_R$ weight is in $[x_i,x_j]$. \label{step:slack}
               \FOR{$i<r<j$} 
                     \STATE Delete $x_r$, coreset $\mathcal{K}_{x_r}$ and partition $\Lambda_{\mathcal{K}_{x_r}}$. Update the indices in sequence $X$ accordingly. \label{delete}
               \ENDFOR     
	\ENDFOR
	\STATE  Let $i$ be the smallest index such that $p_{x_i}$ is expired and $p_{x_{i+1}}$ is active.  \label{expired}
         \FOR {$r<i$}
               \STATE Delete $x_r$ and coreset $\mathcal{K}_{x_r}$ and partition $\Lambda_{\mathcal{K}_{x_r}}$, and update the indices in sequence $X$.
         \ENDFOR
\end{algorithmic}

\textbf{Output Process:}
\begin{algorithmic}[1]
    \STATE  Return coreset $\mathcal{K}_{x_1}$ maintained by {\sc MergeReduce}$([x_1,N]=\{p_{x_1},\dots, p_{N-1}\},p_N)$.
\end{algorithmic}

\caption{{\sc SWCoreset}}
\label{alg:coreset:SW}
\end{algorithm*}

\begin{theorem}
\label{lem:num:coreset:SW}
Let $P\subseteq \REAL^d$ be a point set of size $n$. 
Suppose the optimal cost of clustering of point set $P$ is $\opt_{P}=n^{O(c)}$ for some constant $c$.
Let $s$ be the size of a coreset (constructed using one of the coreset techniques~\cite{HPM04,Ch09})
that  merge-and-reduce method maintains for $P$ in the insertion-only streaming model. 
There exists a sliding window algorithm that maintains this coreset using $O(s^2\epsilon^{-2}\log n)$ space.
\end{theorem}

\begin{proof}
According to Algorithm {\sc SWCoreset}, the next index that we keep in sequence $X$ occurs when $\epsilon$-fraction of 
a region $R \in \Lambda_{\mathcal{K}_{x_j}}$ changes. 
Since $s$ is the size of $(k,\epsilon)$-coreset that  Algorithm {\sc MergeReduce} maintains for $P$, 
the upper bound on the number of regions in partition $\Lambda_{\mathcal{K}_{x_j}}$ is also $s$. 
By Lemma~\ref{lem:coreset:eps:slack}, as long as at most $\epsilon$-fraction of the weight of a region in partition 
$\Lambda_{\mathcal{K}_{x_j}}$ drops, we still have a $(k,\epsilon)$-coreset. Thus, after at most $s/\epsilon$ indices, 
the optimal clustering cost drops by at least $\epsilon$-fraction of its cost. 
Therefore, after $O(\log_{1+\epsilon} n)=O(\epsilon^{-1}\log n)$ of this sequence of $\epsilon$-fraction drops in the optimal clustering cost, 
the cost converges to zero. 
Overall, the number of indices that we maintain is $O(s\epsilon^{-2}\log n)$. 
Moreover, for each index we maintain a $(k,\epsilon)$-coreset of size $s$ using Algorithm {\sc MergeReduce}; 
therefore, the space usage of our algorithm is $O(s^2\epsilon^{-2}\log n)$.
\end{proof}


\newcommand{\Proc}{Proceedings of the~}
\newcommand{\STOC}{Annual ACM Symposium on Theory of Computing (STOC)}
\newcommand{\FOCS}{IEEE Symposium on Foundations of Computer Science (FOCS)}
\newcommand{\SODA}{Annual ACM-SIAM Symposium on Discrete Algorithms (SODA)}
\newcommand{\SOCG}{Annual Symposium on Computational Geometry (SoCG)}
\newcommand{\ICALP}{Annual International Colloquium on Automata, Languages and Programming (ICALP)}
\newcommand{\ESA}{Annual European Symposium on Algorithms (ESA)}
\newcommand{\CCC}{Annual IEEE Conference on Computational Complexity (CCC)}
\newcommand{\RANDOM}{International Workshop on Randomization and Approximation Techniques in Computer Science (RANDOM)}
\newcommand{\PODS}{ACM SIGMOD Symposium on Principles of Database Systems (PODS)}
\newcommand{\SSDBM}{ International Conference on Scientific and Statistical Database Management (SSDBM)}
\newcommand{\ALENEX}{Workshop on Algorithm Engineering and Experiments (ALENEX)}
\newcommand{\BEATCS}{Bulletin of the European Association for Theoretical Computer Science (BEATCS)}
\newcommand{\CCCG}{Canadian Conference on Computational Geometry (CCCG)}
\newcommand{\CIAC}{Italian Conference on Algorithms and Complexity (CIAC)}
\newcommand{\COCOON}{Annual International Computing Combinatorics Conference (COCOON)}
\newcommand{\COLT}{Annual Conference on Learning Theory (COLT)}
\newcommand{\COMPGEOM}{Annual ACM Symposium on Computational Geometry}
\newcommand{\DCGEOM}{Discrete \& Computational Geometry}
\newcommand{\DISC}{International Symposium on Distributed Computing (DISC)}
\newcommand{\ECCC}{Electronic Colloquium on Computational Complexity (ECCC)}
\newcommand{\FSTTCS}{Foundations of Software Technology and Theoretical Computer Science (FSTTCS)}
\newcommand{\ICCCN}{IEEE International Conference on Computer Communications and Networks (ICCCN)}
\newcommand{\ICDCS}{International Conference on Distributed Computing Systems (ICDCS)}
\newcommand{\VLDB}{ International Conference on Very Large Data Bases (VLDB)}
\newcommand{\IJCGA}{International Journal of Computational Geometry and Applications}
\newcommand{\INFOCOM}{IEEE INFOCOM}
\newcommand{\IPCO}{International Integer Programming and Combinatorial Optimization Conference (IPCO)}
\newcommand{\ISAAC}{International Symposium on Algorithms and Computation (ISAAC)}
\newcommand{\ISTCS}{Israel Symposium on Theory of Computing and Systems (ISTCS)}
\newcommand{\JACM}{Journal of the ACM}
\newcommand{\LNCS}{Lecture Notes in Computer Science}
\newcommand{\RSA}{Random Structures and Algorithms}
\newcommand{\SPAA}{Annual ACM Symposium on Parallel Algorithms and Architectures (SPAA)}
\newcommand{\STACS}{Annual Symposium on Theoretical Aspects of Computer Science (STACS)}
\newcommand{\SWAT}{Scandinavian Workshop on Algorithm Theory (SWAT)}
\newcommand{\TALG}{ACM Transactions on Algorithms}
\newcommand{\UAI}{Conference on Uncertainty in Artificial Intelligence (UAI)}
\newcommand{\WADS}{Workshop on Algorithms and Data Structures (WADS)}
\newcommand{\SICOMP}{SIAM Journal on Computing}
\newcommand{\JCSS}{Journal of Computer and System Sciences}
\newcommand{\JASIS}{Journal of the American society for information science}
\newcommand{\PMS}{ Philosophical Magazine Series}
\newcommand{\ML}{Machine Learning}
\newcommand{\DCG}{Discrete and Computational Geometry}
\newcommand{\TODS}{ACM Transactions on Database Systems (TODS)}
\newcommand{\PHREV}{Physical Review E}
\newcommand{\NATS}{National Academy of Sciences}
\newcommand{\MPHy}{Reviews of Modern Physics}
\newcommand{\NRG}{Nature Reviews : Genetics}
\newcommand{\BullAMS}{Bulletin (New Series) of the American Mathematical Society}
\newcommand{\AMSM}{The American Mathematical Monthly}
\newcommand{\JAM}{SIAM Journal on Applied Mathematics}
\newcommand{\JDM}{SIAM Journal of  Discrete Math}
\newcommand{\JASM}{Journal of the American Statistical Association}
\newcommand{\AMS}{Annals of Mathematical Statistics}
\newcommand{\JALG}{Journal of Algorithms}
\newcommand{\TIT}{IEEE Transactions on Information Theory}
\newcommand{\CM}{Contemporary Mathematics}
\newcommand{\JC}{Journal of Complexity}
\newcommand{\TSE}{IEEE Transactions on Software Engineering}
\newcommand{\TNDE}{IEEE Transactions on Knowledge and Data Engineering}
\newcommand{\JIC}{Journal Information and Computation}
\newcommand{\ToC}{Theory of Computing}
\newcommand{\MST}{Mathematical Systems Theory}
\newcommand{\Com}{Combinatorica}
\newcommand{\NC}{Neural Computation}
\newcommand{\TAP}{The Annals of Probability}
\newcommand{\TCS}{Theoretical Computer Science}

\bibliographystyle{plain}
\bibliography{4-References}
\newpage

\appendix

\section{Missing Proofs and Further Results from Metric $k$-Median Clustering in Sliding Windows}
\label{appendix:metric}

\subsection{$k$-median and $k$-means are not smooth functions}
\label{app:proof:notsmooth}

\begin{claim}
$k$-median and $k$-means clustering are not smooth functions.
That is, there exist sets of points $A, B$ and $C$
such that for any $\gamma, \beta > 0$,
$\opt(A \cup B,k) \le \gamma \opt(B,k)$
but $\opt(A \cup B \cup C) > \beta \opt(B \cup C)$.
\end{claim}
\begin{proof}
Let $A\cup B$ consist of $k$ distinct points, with $A \neq \emptyset$,
$B \neq \emptyset$ and $A$ contains one or more points not in $B$.
Then $B$ contains at most $k-1$ distinct points, and
$\opt(A \cup B,k) = 0 \le \gamma \opt(B,k) = 0$ for any $\gamma$.
Consider a set $C$ consisting of a single point
and satisfying $C \cap (A \cup B) = \emptyset$.
Then $A \cup B \cup C$ has at least $k+1$ distinct points,
so that $\opt(A \cup B \cup C,k) > 0$,
while $\opt(B \cup C,k)=0$.
Then for any $\beta$, we have that
$\opt(A \cup B \cup C,k) > 0 = \beta \opt(B \cup C,k)$.
\end{proof}

\subsection{Proof of Lemma~\ref{bounding}}
\label{app:proof:bounding}

\begin{proofof}{Lemma~\ref{bounding}}
Let $t$ be the $\beta$-approximate clustering map of $A\cup B$
in the hypothesis.
By assumption, $t$ induces partitions $A_1,A_2,\dots,A_k$ and
$B_1,B_2,\dots,B_k$ of $A$ and $B$, respectively, given by
$A_i = t^{-1}(c_i)\cap A$ and $B_i = t^{-1}(c_i)\cap B$.
Since $|A_i|\le |B_i|$ for all $i \in [k]$ by assumption,
for each $i \in [k]$ there exists a one-to-one mapping
$g_i$ from $A_i$ to a subset of $B_i$.
Letting $a\in A_i$, we have
\begin{equation*}
\tri^{-1} d(a,g_i(a)) \le d(a,t(a)) + d(g_i(a),t(a))
        = d(a,t(a)) + d(g_i(a),t(g_i(a))),
\end{equation*}
where the first inequality is the approximate triangle inequality
and the second inequality follows from the fact that
$t(g_i(a))= c_i = t(a)$ by definition of $g_i$. Thus,
\begin{align}
\tri^{-1} \sum_{i=1}^k\sum_{a\in A_i} d(a,g_i(a))
  &\le \sum_{i=1}^k \sum_{a\in A_i}
        \big[d(a,t(a)) + d(g_i(a),t(g_i(a)))\big] \nonumber \\
  &\le \sum_{x\in A \cup B} d(x,t(x))
   \le \beta \opt(A\cup B,k) \le \beta \gamma \opt(B,k) \label{eq:attachAs_1},
\end{align}
where the first inequality follows from the approximate triangle inequality,
the second inequality follows from the definition of $g_i$,
the third inequality follows since $t$ is $\beta$-approximate,
and the fourth inequality holds by assumption.
Let $s$ be an optimal clustering map for $B\cup C$
(i.e., $s$ is $1$-approximate).  Then we have
\begin{equation}\label{eq:attachAs_2}
\sum_{i=1}^k\sum_{a\in A_i} d(g_i(a),s(g_i(a)))
\le \sum_{i=1}^k\sum_{b\in B_i} d(b,s(b))
\le \opt(B \cup C,k).
\end{equation}
Bounding the cost of connecting $A$ to the optimal centers of $B\cup C$,
we obtain
\begin{align*}
\sum_{a\in A}  d(a,s(a))
  &= \sum_{i=1}^k \sum_{a\in A_i}  d(a,s(a))
        \le \sum_{i=1}^k \sum_{a\in A_i}  d(a,s(g_i(a))) \\
  &\le \sum_{i=1}^k \sum_{a\in A_i}
        \tri \big[ d(a,g_i(a)) + d(g_i(a),s(g_i(a))) \big] 
  \le \tri \opt(B \cup C,k) + \beta \gamma \tri \opt(B,k),
\end{align*}
where the first inequality follows from the fact that $s(a)$
is the closest center to $a$ by definition,
the second inequality follows from the approximate triangle inequality,
and the third inequality follows from equations~\eqref{eq:attachAs_1}
and~\eqref{eq:attachAs_2}.
Thus
we conclude that
\begin{align*}
\opt(A\cup B \cup C,k)
  &\le \sum_{a\in A} d(a,s(a)) \enspace + \sum_{x\in B\cup C} d(x,s(x)) \\
  &\le (1+\tri) \opt(B \cup C,k) + \beta \gamma \tri \opt(B,k)
  \le (1+\tri+\beta \gamma\tri)\opt(B\cup C,k).
\end{align*}
\end{proofof}

\subsection{Proof of Lemma~\ref{maintainSizes}}
\label{app:proof:maintainSizes}

\begin{proofof}{Lemma~\ref{maintainSizes}}
Using our algorithm from Lemma~\ref{blackbox},
we proceed as before until we are notified of the first point of $A_2$.
Here, we store $\mathcal{S}_{1,2} \gets \PLS(A_2)$,
but we continue running this instance of the \textPLS algorithm.
As before, we commence a new instance of the \textPLS algorithm
beginning with the first point of $A_2$.
In general, whenever a transition occurs to $A_j$,
for all $i < j$ we store
$\mathcal{S}_{i,j} \gets \PLS(A_i \cup \dots \cup A_{j-1})$
and continue running all instances.
As a result, we maintain sets $\mathcal{S}_{i,j}$ for $i \in [Z]$ and $j > i$.
There are $O(Z^2)$ such sets, each of size $O(k \log^2 n)$.
When we wish to compute the $\beta$-approximate map, we run an offline
$O(1)$-approximation on $\mathcal{S}_{j,\ell} \cup \mathcal{S}_{\ell , Z}$.
The cluster sizes are computed as in Lemma~\ref{blackbox}.
\end{proofof}

\subsection{Proof of Lemma~\ref{approxCost} }
\label{app:proof:approxCost}

\begin{proofof}{Lemma~\ref{approxCost} }
If $X_s = m$, we have equality, so suppose $X_s < m$,
implying that index $m$ was previously deleted
at some time $Q$, when $p_Q$ the most recent point to arrive.
Let $z$ be the index (assigned before deletion) such that $X_z = m$.
During some iteration of the loop on Lines~\ref{dropOPT}-\ref{endDropOPT},
it must have held after Line~\ref{assignJ} that $X_i \le X_s < X_z < X_j$.
This is because both $X_i$ and $X_j$ are stored,
so $s \ge i$ by maximality of $s$.

Line~\ref{assignJ} guarantees $\beta R([X_i,Q]) \le \gamma R([X_j,Q])$.
Since the $\beta$-approximation ensures that
$\opt(\cdot) \le R(\cdot) \le \beta \opt(\cdot)$,
this implies that $\opt([X_i,Q]) \le \gamma \opt([X_z,Q])$.
Let $t$ denote the (uncalculated, but existent)
$\beta$-approximate map of $[X_i,Q]$ as in Lemma~\ref{blackbox}.
The loop on Line~\ref{dropsize} ensures that
$|t^{-1}(c_w) \cap S_{i,z}| \le |t^{-1}(c_w) \cap S_{i,T}|$
for every $w \in [k]$.
Therefore, Lemma~\ref{bounding} guarantees that for any suffix $C$,
$\opt([X_i,Q]\cup C) \le (2+\beta \gamma) \opt([m,Q]\cup C)$.
By letting $C = [Q+1, N]$, the result is obtained.
\end{proofof}

\subsection{Proof of Lemma~\ref{updateTime}}
\label{app:proof:updateTime}
\begin{proofof}{Lemma~\ref{updateTime}}
Let $h(\cdot)$ be the update time for \textPLS
and $g(\cdot)$ be the time for the offline $c$-approximation.
Feeding new point $p_N$ to all $T$ instances of \textPLS
requires $T h(n)$ time
and computing the $c$-approximation for all $O(\log \opt) = O(\log n)$
iterations of the loop on Line~\ref{dropOPT} requires
$O(\log n \cdot g(k \log^2 n))$ time.
Partitioning each bucket requires $O(T k^2 \log^2 n)$ time,
and finding the maximal index on Line~\ref{getMax} requires $O(T^2 k)$ time.
In total, an update takes
$O(h(n) \cdot T + \log n \cdot g(k \log^2 n) + Tk^2\log^2 n + T^2 k)$ time.
By Lemma~\ref{spaceBound}, $T = O(k \log^2 n)$.
By~\cite{COP03}, the update time of \textPLS is polynomial in its argument,
and using any of a number of offline $O(1)$-approximations for
$k$-median, for example,~\cite{AGKMMP04,JV1999}.
Moreover, $g(\cdot)$ and $h(\cdot)$
are such that the last two terms are the largest factors,
resulting in an update time of $O(k^3 \log^4 n)$.
\end{proofof}

\subsection{Algorithm for Metric {\sc $k$-means}}
The metric $k$-means problem is defined similarly. 

\begin{definition}[Metric $k$-means]
Let $P \subseteq X$ be a set of $n$ points in a metric space $(X,d)$
and let $k\in \NATURAL$ be a natural number.
Suppose $C=\{c_1,\dots,c_k\}$ is a set of $k$ centers.
The clustering of point set $P$ using $C$ is the partitioning of $P$ such that
a point $p\in P$ is in cluster $C_i$ if $c_i\in C$ is the nearest center to $p$ in $C$,
that is point $p$ is assigned to its nearest center $c_i\in C$.
The cost of $k$-median clustering by $C$ is
$\cost^2(P,C)=\sum_{p\in P} \dist^2(p,C). $
The metric $k$-median problem is to find a set $C\subset P$ of $k$ centers that minimizes
the cost $\cost^2(P,C)$, that is
\[
   \begin{split}
        \cost^2(P,C)&=\sum_{p\in P} \dist^2(p,C)=\min_{C'\subset P: |C'|=k} \cost^2(P,C')\\
        &=\min_{C'\subset P: |C'|=k} \sum_{p\in P} \dist^2(p,C') \enspace ,
   \end{split}
\]
where $\dist^2(p,C)=\min_{c\in C} \dist^2(p,c)$ and $\dist^2(p,C')=\min_{c\in C'} \dist^2(p,c)$
\end{definition}

A concept used by our algorithm for the metric $k$-means is the notion of $\sigma$-separability \cite{BMOR11}. 
Intuitively, data which is separable for a high value of $\sigma$ is well-clusterable into $k$-clusters 
(i.e. removing one center greatly increases the optimal cost).

\begin{definition}[$\sigma$-separable dataset] \cite{BMOR11}
\label{separability}
A set of input data is said to be $\sigma$-separable if the ratio of the optimal $k$-means cost to the optimal $(k-1)$-means cost is at most $\sigma^2$.
\end{definition}

We now turn to the modifications necessary for {\sc $k$-means}.  We state the main theorem and explain the necessary changes in the remainder of this section.

\begin{theorem} \label{metricKMeans}
Assuming $\opt' = poly(n)$ and every window is $\sigma$-separable for some $\sigma = O(1)$, there exists a sliding windows algorithm which maintains a $O(1)$-approximation for the metric {\sc $k$-median} problem using $O(k^3 \log^6 n)$ space and $O(k^3 \log^4 n)$ update time.
\end{theorem}

In Lemma~\ref{blackbox}, for {\sc $k$-median} we had used the \textbf{PLS} algorithm to maintain a weighted set $\mathcal{S}$
such that $\cost(\mathcal{P},\mathcal{S}) \le \alpha \opt(\mathcal{P},k)$
for some constant $\alpha$.
Instead, we now use the insertion-only {\sc $k$-means} algorithm of \cite{BMOR11}.  This algorithm also works by providing a weighted set $\mathcal{S}$
such that $\cost(\mathcal{P},\mathcal{S}) \le \alpha \opt(\mathcal{P},k)$ 
for some constant $\alpha$.  Here, the space required is again $O(k \log^2 n)$.  The approximation-factor $\alpha$ is now $O(\sigma^2)$ where the data is $\sigma$-separable as defined in Definition~\ref{separability}.
The second modification is in Algorithm~\ref{alg:kmedian:approx:SW}.  For {\sc $k$-median}, we had used Lemma~\ref{bounding} with $\lambda = 1$ since the {\sc $k$-median} function satisified the $1$-approximate triangle inequality (i.e. the standard triangle inequality).  For {\sc $k$-means}, we now satisfy the $2$-approximate triangle inequality, so we use the lemma with $\lambda = 2$.  Theorem~\ref{mainTheorem} still holds without modification, so we result in a $O(\sigma^4)$-approximation.


\section{Missing Proofs and Further Results from Euclidean Coresets in Sliding Windows}
\label{appendix:euclid}

Here, we briefly review the definition of VC-dimension and $\epsilon$-sample
as presented in Alon and Spencer's book \cite{as-pm}.
A \emph{range space} $S$ is a pair $(X,R)$,
where $X$ is a set and $R$ is a family of subsets of $X$.
The elements of $X$ are called \emph{points}
and the subsets in $R$ are called \emph{ranges}.
For $A \subseteq X$, we define the \emph{projection} of $R$ on $A$ as
$P_R(A)=\{r\cap A: r\in R\}$.
We say $A$ is \emph{shattered} if $P_R(A)$ contains all subsets of $A$.
The \emph{Vapnik-Chervonenkis} dimension (VC-D) of $S$,
which we denote by $d_{VC}(S)$, is the maximum cardinality
of a shattered subset of $X$.
If there exist arbitrarily large shattered subsets of $X$,
then $d_{VC}(S)=\infty$.
Let $(X,R)$ be a range space with VC-D $d$,
$A \subseteq X$ with $|A|$ finite, and let $0<\epsilon<1$ be a parameter.
We say that $B \subset A$ is an \emph{$\epsilon$-sample} for $A$
if for any range $r\in R$,
$|\frac{|A \cap r|}{|A|}-\frac{|B\cap r|}{|B|}|\le \epsilon$.

\begin{lemma}[\cite{as-pm}]
\label{lem:vc:dim}
Let $S=(X,R)$ be a range space of VC-D $d_{VC}(S)$,
$A \subseteq X$ with $|A|$ finite, and let $0<\epsilon<1$ be a parameter.
Let $c>1$ be a constant and $0<\delta<1$.
Then a random subset $B$ of $A$ of size
$s=\min(|A|,\frac{c}{\epsilon^2}\cdot (d_{VC}(S)\log(d_{VC}(S)/\epsilon)+\log(1/\delta)))$
is an $\epsilon$-sample for $A$ with probability at least $1-\delta$.
 \end{lemma}

 \begin{lemma}[\cite{H11} (Chapter 5)]
 \label{lem:mix:vc:dim}
 Let $S=(X,R)$ and $T=(X,R')$ be two range spaces of VC-D $d_{VC}(S)$ and $d'_{VC}(T)$,
 respectively, where $d_{VC}(S),d'_{VC}(T)>1$. 
 Let $U=\{r\cup r' | r \in R , r'\in R'\}$ and
  $I=\{r\cup r' | r \in R , r'\in R'\}$. 
 Then the range spaces $\hat{S}=(X,U)$ and $\hat{S}'=(X,I)$
 have $d_{VC}(\hat{S}) = d_{VC}(\hat{S}')=O((d_{VC}(S)+d'_{VC}(T))\log(d_{VC}(S)+d'_{VC}(T)))$.
 In general, unions, intersections and
any finite sequence of combining ranges spaces with finite
 VC-D results in a range space with a finite VC-D.
 \end{lemma}

In $\REAL^d$, the VC-D of a half space is $d$~\cite{H11}. 
Balls, ellipsoids and cubes in $\REAL^d$ have VC-D $O(d^2)$ as can be easily verified 
by lifting the points into $O(d^2)$ dimensions, where each one of these shapes in the original space is mapped into a half space.

Next, we first review $2$ well-known coreset techniques that fit into the framework of Section \ref{sec:euclid}. 
These coreset techniques are\\ 
(1) Coreset technique of  \cite{HPM04} due to Har-Peled and Mazumdar for the $k$-median and the $k$-means in low dimensional Euclidean spaces, 
i.e., when dimension $d$ is constant.\\
(2) Coreset technique of  \cite{Ch09} due to Chen for the $k$-median and the $k$-means problems in high dimensional Euclidean  spaces, 
i.e., when dimension $d$ is not constant.\\

Next, we prove Lemma \ref{lem:coreset:eps:slack} for each one of these coreset techniques. 
Interestingly, the coreset technique of  \cite{HPM04} is the basis of almost all follow-up coresets 
for the $k$-median and the $k$-means in low dimensional Euclidean spaces. 
This includes the coreset techniques of \cite{FFS06,HPK05}.
The same is true for the coreset technique of  \cite{Ch09} which is the basis of almost all follow-up coresets 
for the $k$-median and the $k$-means in high dimensional Euclidean spaces. 
This includes the coreset techniques of \cite{FMS07,FMSW10,M10,FL11}. 

Finally we prove Lemma \ref{lem:bucket:sum:small:error}.
%


\subsection{Coreset Technique of  \cite{HPM04} due to Har-Peled and Mazumdar}
\label{sec:HM:coreset}
We explain their coreset technique for the $k$-median problem. 
The same coreset technique works for the $k$-means problem. 
We first invoke a $\alpha$-approximation algorithm on a point set $P\in \REAL^d$ that returns  
a set $C=\{c_1,\cdots,c_k\}\subset \REAL^d$ of $k$ centers. Let $C_i$ be the set of points that 
are in the cluster of center $c_i$, i.e. $C_i=\{p\in P: \dist(p,c_i)\le \min_{c_j\in C}\dist(p,c_j)\}$. 
We consider $\log n+1$ balls $\text{Ball}_{i,j}(c_i,2^j\cdot \frac{\cost(P,C)}{n})$ centered at center $c_i$ 
of radii $2^j\cdot \frac{\cost(P,C)}{n}$ for $j\in \{0,1,2,\cdots,\log n\}$. 
In $\text{Ball}_{i,j}$, we impose a grid $G_{i,j}$ of side length $\frac{\epsilon}{10\sqrt{d}\alpha}\cdot 2^j\cdot \frac{\cost(P,C)}{n}$ 
and for every non-empty cell $c$ in grid $G_{i,j}$ we replace all points in $c$ with one 
point of weight $n_{c}$, where $n_{c}$ is the number of points in $c$. 
Let $\mathcal{K}_P$ be the set of cells returned by this coreset technique, 
that is  $\mathcal{K}_P=\cup_{c\in G_{i,j}} c$. 
We also let partition $\Lambda_P$ to be the set of cells $\Lambda_P=\{G_{i,j} \cap \text{Ball}_{i,j}: i \in [k], j\in [\log n+1]\}$.


In Step 3 of Algorithm \ref{alg:unify:coreset} we let 
$s_{\mathcal{CC}}=f(n,d,\epsilon,\delta)=1$ (which is the number of points that this coreset technique 
samples from each cell) be  $1$. Moreover,  we let parameter $s$ in Lemma \ref{lem:num:coreset:SW} (which is 
the size of this coreset maintained by merge-and-reduce approach) be $O(k \epsilon^{-d} \log^{2d+2}{n})$.
We now prove a variant of Lemma \ref{lem:coreset:eps:slack} for this coreset. 
Observe that to adjust the parameter $\epsilon$ in Lemma \ref{lem:coreset:eps:slack} 
for Lemma \ref{lem:slack:k-median:low}, we replace $\epsilon$ by $\frac{\epsilon^2}{5\sqrt{d}}$. 

\begin{lemma}
\label{lem:slack:k-median:low}
Let $P\subset \REAL^d$ be a point set of size $n$ and $k\in\NATURAL$ be a parameter. 
Let $\Lambda_P=\{c_1,c_2,\cdots,c_x\}$ be the partition set of cells returned by the coreset technique 
of Har-Peled and Mazumdar \cite{HPM04}. 
Let $\mathcal{B}=\{b_1,\cdots,b_k\}\subset \REAL^d$ be an arbitrary set of $k$ centers. 
Suppose for every cell $c \in \Lambda_P$ we delete up to $\frac{\epsilon^2}{5\sqrt{d}}\cdot n_{c}$ points 
and let  $c'$ be cell $c$ after deletion of these points. We then have 
\[
\sum_{c\in \Lambda_P} |\cost(c',\mathcal{B}) - \cost(c,\mathcal{B})| \le 
\epsilon\cdot \cost(P,\mathcal{B}) \enspace ,
\]
where $\cost(c,\mathcal{B})=\sum_{p\in c} \dist(p,\mathcal{B})$ and $\cost(c',\mathcal{B})=\sum_{p\in c'} \dist(p,\mathcal{B})$. 
\end{lemma}

\begin{proof}
Let us fix a cell $c\in \Lambda_P$. Suppose $c\in G_{i,j}$ for a particular cluster $P_i$ and $j\in \{0,1,2,\cdots,\log n\}$. 
Assume the nearest center to cell $c$ is $b_{c}\in \mathcal{B}$. 
We either have $\dist(c,b_{c})\ge \frac{\ell_{c}}{\epsilon}$ or $\dist(c,b_{c})< \frac{\ell_{c}}{\epsilon}$. 
If we have $\dist(c,b_{c})\ge \frac{\ell_{c}}{\epsilon}$, then 
$\cost(c,\mathcal{B})\ge n_{c}\cdot \dist(c,b_{c})\ge n_{c}\cdot \frac{\ell_{c}}{\epsilon}$. 
On the other, since we delete up to $\frac{\epsilon^2}{\sqrt{d}}\cdot n_{c}$ points of cell $c$, 
the cost that we lose is at most 
\[
  \frac{\epsilon^2}{\sqrt{d}}\cdot n_{c} (\sqrt{d}\ell_{c}+\dist(c,b_{c})) 
  \le \frac{\epsilon^2}{\sqrt{d}}\cdot n_{c}\cdot (\sqrt{d}\epsilon+1)\cdot \dist(c,b_{c}) 
  \le \epsilon\cdot \cost(c,\mathcal{B}) \enspace .
\]

Now suppose we have $\dist(c,b_{c})< \frac{\ell_{c}}{\epsilon}$. 
We have two cases, either $j=0$ or $j>0$. 
We first prove the lemma when $j=0$. 
Since cell $c$ is in $G_{i,0}$, we have $\ell_c=\frac{\epsilon}{10\sqrt{d}\alpha}\cdot \frac{\cost(P,C)}{n}$. 
Therefore, $\dist(c,b_{c})< \frac{\ell_{c}}{\epsilon}\le \frac{1}{10\sqrt{d}\alpha}\cdot \frac{\cost(P,C)}{n}$.
From every cell in $Ball_{i,0}\cap G_{i,0}$ we delete up to $\frac{\epsilon^2}{\sqrt{d}}\cdot n_{c}$ points. 
Therefore, the cost that we lose is at most 
\[
  \sum_{c\in Ball_{i,0}\cap G_{i,0}} \frac{\epsilon^2}{\sqrt{d}}\cdot n_{c} \cdot \frac{1}{10\sqrt{d}\alpha}\cdot \frac{\cost(P,C)}{n} 
  \le \frac{\epsilon^2}{10d\alpha}\cdot \cost(P,C) \le \frac{\epsilon^2}{10d}\cdot \cost(P,\mathcal{B})\enspace .
\]

The second case is when $j>0$. Recall that $\dist(c,b_{c})< \frac{\ell_{c}}{\epsilon}$. 
Observe that since $j>0$, $\dist(c,c_i)\ge 2^{j-1}\cdot \frac{\cost(P,C)}{n}$ and 
$\ell_c=\frac{\epsilon}{10\sqrt{d}\alpha}\cdot 2^j\cdot \frac{\cost(P,C)}{n}\le \frac{\epsilon}{5\sqrt{d}\alpha}\cdot \dist(c,c_i)$ 
which means $\dist(c,b_{c})< \frac{\ell_{c}}{\epsilon}\le \frac{1}{5\sqrt{d}\alpha}\cdot \dist(c,c_i)$. 
For each such cell $c$ we delete up to $\frac{\epsilon^2}{\sqrt{d}}\cdot n_{c}$ points. 
Thus, taking the summation over all $i$ and $j>0$ yields 
\[
  \begin{split}
  \sum_{i=1}^k \sum_{j>0} \sum_{ c\in G_{i,j}\cap Ball_{i,j}} \frac{\epsilon^2}{\sqrt{d}}\cdot n_{c}\cdot \dist(c,b_{c})
  &\le \sum_{i=1}^k \sum_{j>0} \sum_{ c\in G_{i,j}\cap Ball_{i,j}} \frac{\epsilon^2}{\sqrt{d}}\cdot n_{c}\cdot \frac{1}{5\sqrt{d}\alpha}\cdot \dist(c,c_i)\\
   &\le \frac{\epsilon^2}{5d\alpha}\cdot \cost(P,C) \le \frac{\epsilon^2}{5d}\cdot \cost(P,\mathcal{B})\enspace .
  \end{split}
\]
\end{proof}


\subsection{Coreset Technique of  \cite{Ch09} due to Chen}
\label{sec:chen:coreset}
We explain his coreset technique for the $k$-median problem. 
The same coreset technique works for the $k$-means problem. 
We first invoke a $\alpha$-approximation algorithm on a point set $P\in \REAL^d$ that returns  
a set $C=\{c_1,\cdots,c_k\}\subset \REAL^d$ of $k$ centers. Let $C_i$ be the set of points that 
are in the cluster of center $c_i$, i.e. $C_i=\{p\in P: \dist(p,c_i)\le \min_{c_j\in C}\dist(p,c_j)\}$. 
We consider $\log n+1$ balls $\text{Ball}_{i,j}(c_i,2^j\cdot \frac{\cost(P,C)}{n})$ centered at center $c_i$ 
of radii $2^j\cdot \frac{\cost(P,C)}{n}$ for $j\in \{0,1,2,\cdots,\log n\}$. 
In ring $\text{Ball}_{i,j}\backslash \text{Ball}_{i,j-1}$ having $n_{i,j}$ points, 
we take a sample set $S_{i,j}$ of size $s=\min(n_{i,j},\epsilon^{-2}dk\log(\frac{k\log n}{\epsilon\delta})$ points uniformly at random 
and we assign a weight of  $n_{i,j}/s$ to each sampled point. 
We replace all points in $\text{Ball}_{i,j}$ with weighted set $S_{i,j}$. 
Let $\mathcal{K}_P$ be the union set of weighted sampled sets returned by this coreset technique, 
that is  $\mathcal{K}_P=\cup_{i,j} S_{i,j}$. 
We also let partition $\Lambda_P$ to be the set of rings $\Lambda_P=\{\text{Ball}_{i,j}\backslash \text{Ball}_{i,j-1}: i \in [k], j\in [\log n+1]\}$.


In Step 3 of Algorithm \ref{alg:unify:coreset} for each ring $\text{Ball}_{i,j}\backslash \text{Ball}_{i,j-1}$ having $n_{i,j}$ points we let 
$s_{\mathcal{CC}}=f(n,d,\epsilon,\delta)$  be  $\min(n_{i,j},\epsilon^{-2}dk\log(\frac{k\log n}{\epsilon\delta})$. 
Moreover,  we let parameter $s$ in Lemma \ref{lem:num:coreset:SW} (which is 
the size of this coreset maintained by merge-and-reduce approach) be $O(k^2d\epsilon^{-2}\log^8 n)$.
We now prove a variant of Lemma \ref{lem:coreset:eps:slack} for this coreset. 
Observe that to adjust the parameter $\epsilon$ in Lemma \ref{lem:coreset:eps:slack} 
for Lemma \ref{lem:slack:k-median:low}, we replace $\epsilon$ by $\frac{\epsilon^2}{5\sqrt{d}}$.

\begin{lemma}
\label{lem:slack:k-median:high}
Let $P\subset \REAL^d$ be a point set of size $n$ and $k\in\NATURAL$ be a parameter. 
Let $\Lambda_P=\{\text{Ball}_{i,j}: i \in [k], j\in [\log n+1]\}$ be the partition set of balls returned by the coreset technique 
of Chen \cite{Ch09}. 
Let $\mathcal{B}=\{b_1,\cdots,b_k\}\subset \REAL^d$ be an arbitrary set of $k$ centers. 
Suppose for every ring $\text{Ball}_{i,j}\backslash \text{Ball}_{i,j-1} \in \Lambda_P$ having $n_{i,j}$ points, we delete up to $\frac{\epsilon^2}{\sqrt{d}}\cdot n_{i,j}$ points 
and let  $(\text{Ball}_{i,j}\backslash \text{Ball}_{i,j-1})'$ be this ball after deletion of these points. We then have 
\[
\sum_{\text{Ball}_{i,j}\backslash \text{Ball}_{i,j-1}\in \Lambda_P} |\cost((\text{Ball}_{i,j}\backslash \text{Ball}_{i,j-1})',\mathcal{B}) - \cost(\text{Ball}_{i,j}\backslash \text{Ball}_{i,j-1},\mathcal{B})| \le 
\epsilon\cdot \cost(P,\mathcal{B}) \enspace ,
\]
where $\cost(\text{Ball}_{i,j}\backslash \text{Ball}_{i,j-1},\mathcal{B})=\sum_{p\in \text{Ball}_{i,j}\backslash \text{Ball}_{i,j-1}} \dist(p,\mathcal{B})$ and 
$\cost((\text{Ball}_{i,j}\backslash \text{Ball}_{i,j-1})',\mathcal{B})=\sum_{p\in (\text{Ball}_{i,j}\backslash \text{Ball}_{i,j-1})'} \dist(p,\mathcal{B})$. 
\end{lemma}

\begin{proof}
The proof is in the same spirit of the proof of Lemma \ref{lem:slack:k-median:low}. 
Let us fix a ring $\text{Ball}_{i,j}\backslash \text{Ball}_{i,j-1} \in \Lambda_P$ for $i\in [k]$ and $j\in \{0,1,\cdots, \log n \}$. 
Observe that the radius of $\text{Ball}_{i,j}$ is $2^j\cdot \frac{\cost(P,C)}{n}$ and the radius of $\text{Ball}_{i,j-1}$ is $2^{j-1}\cdot \frac{\cost(P,C)}{n}$. 
For the simplicity let us denote $\text{Ball}_{i,j}\backslash \text{Ball}_{i,j-1}$ by $R_{i,j}$ 
and we let $\ell_{R_{i,j}}=2^j\cdot \frac{\cost(P,C)}{n}$ and let $n_{R_{i,j}}$ be the number of points in the ring $R_{i,j}$. 
Assume the nearest center to ring $R_{i,j}$ is $b_{R_{i,j}}\in \mathcal{B}$. 
We either have $\dist(R_{i,j},b_{R_{i,j}})\ge \frac{\ell_{R_{i,j}}}{\epsilon}$ or $\dist(R_{i,j},b_{R_{i,j}})< \frac{\ell_{R_{i,j}}}{\epsilon}$. 
If we have $\dist(R_{i,j},b_{R_{i,j}})\ge \frac{\ell_{R_{i,j}}}{\epsilon}$, then 
$\cost(R_{i,j},\mathcal{B})\ge n_{R_{i,j}}\cdot \dist(R_{i,j},b_{R_{i,j}})\ge n_{R_{i,j}}\cdot \frac{\ell_{R_{i,j}}}{\epsilon}$. 
On the other, since we delete up to $\frac{\epsilon^2}{2}\cdot n_{R_{i,j}}$ points of cell $R_{i,j}$, 
the cost that we lose is at most 
\[
  \frac{\epsilon^2}{2}\cdot n_{R_{i,j}} (2\ell_{R_{i,j}}+\dist(R_{i,j},b_{R_{i,j}})) 
  \le \frac{\epsilon^2}{2}\cdot n_{R_{i,j}}\cdot (2\epsilon+1)\cdot \dist(R_{i,j},b_{R_{i,j}}) 
  \le \epsilon\cdot \cost(R_{i,j},\mathcal{B}) \enspace .
\]

Now suppose we have $\dist(R_{i,j},b_{R_{i,j}})< \frac{\ell_{R_{i,j}}}{\epsilon}$. 
We have two cases, either $j=0$ or $j>0$. 
We first prove the lemma when $j=0$. 
For $j=0$, ring $R_{i,0}$ is in fact ball $\text{Ball}_{i,0}$ of radius $\ell_{R_{i,0}}=\frac{\cost(P,C)}{n}$.
Since cell $c$ is in $G_{i,0}$, we have $\ell_c=\frac{\epsilon}{10\sqrt{d}\alpha}\cdot \frac{\cost(P,C)}{n}$. 
Therefore, $\dist(R_{i,0},b_{R_{i,0}})< \frac{\ell_{R_{i,0}}}{\epsilon}\le \frac{1}{\epsilon}\cdot \frac{\cost(P,C)}{n}$.
We delete up to $\frac{\epsilon^2}{2}\cdot n_{R_{i,0}}$ points from $R_{i,0}$. 
Therefore, the cost that we lose is at most 
\[
  \frac{\epsilon^2}{2}\cdot n_{R_{i,0}} \cdot  \dist(R_{i,0},b_{R_{i,0}}) 
  \le \frac{\epsilon^2}{2}\cdot n_{R_{i,0}} \cdot \frac{1}{\epsilon}\cdot \frac{\cost(P,C)}{n} 
  \le \frac{\epsilon}{2}\cdot n_{R_{i,0}} \cdot \frac{\cost(P,C)}{n} \enspace .
\]

We have $i\in [k]$. Hence a summation over $i$ will find the overall cost that we lose as follows. 
\[
  \sum_{i\in [k]} \frac{\epsilon}{2}\cdot n_{R_{i,0}} \cdot \frac{\cost(P,C)}{n} 
  \le \frac{\epsilon}{2}\cdot \cost(P,C)\enspace .
\]

The second case is when $j>0$. Recall that $\dist(R_{i,j},b_{R_{i,j}})< \frac{\ell_{R_{i,j}}}{\epsilon}$. 
Observe that since $j>0$, $\dist(R_{i,j},c_i)\ge 2^{j-1}\cdot \frac{\cost(P,C)}{n}$ and 
$\ell_{R_{i,j}}= 2^j\cdot \frac{\cost(P,C)}{n}\le 2 \dist(R_{i,j},c_i)$ 
which means $\dist(R_{i,j},b_{R_{i,j}})< \frac{\ell_{R_{i,j}}}{\epsilon}\le \frac{2}{\epsilon}\cdot \dist(R_{i,j},c_i)$. 
For each such ring $R_{i,j}$ we delete up to $\frac{\epsilon^2}{2}\cdot n_{R_{i,j}}$ points. 
Thus, taking the summation over all $i$ and $j>0$ yields 
\[
  \begin{split}
  \sum_{i=1}^k \sum_{j>0} \sum_{ R_{i,j}\in \Lambda_P} \frac{\epsilon^2}{2}\cdot n_{R_{i,j}}\cdot \dist(R_{i,j},b_{R_{i,j}})
  &\le \sum_{i=1}^k \sum_{j>0} \sum_{R_{i,j}\in \Lambda_P} \frac{\epsilon^2}{2}\cdot n_{R_{i,j}}\cdot \frac{2}{\epsilon}\cdot \dist(R_{i,j},c_i)\\
   &\le \epsilon\cdot \cost(P,C) \enspace .
  \end{split}
\]

\end{proof}

\COMMENTED{
\subsection{Coreset Technique of  \cite{FMSW10} due to Feldman, Monemizadeh, Sohler and Woodruff}
\label{sec:FMSW:coreset}
The coreset technique of \cite{FMSW10} is for $k=1$ $j$-(dimensional) subspace problem when $q=1$, that is 
given a point set $P\in \REAL^d$ we want to find a $j$-dimensional subspace $F$ such that the sum of distances 
to $F$ is minimum. Here we assume $|P|=n$ and dimension $d$ is not constant. The idea is in fact an extension 
of the coreset technique of \cite{FMS07} due to Feldman, Monemizadeh and Sohler for the $k$-means problem. 
Indeed, the coreset technique of \cite{FMS07} is also an extension of coreset technique of  \cite{Ch09} due to Chen 
for the $k$-median and $k$-means problems in high dimensional Euclidean spaces. The basic idea is to do the random sampling 
non-uniformly to hit points which are far from approximate centers. More precisely, the idea behind this approach 
for the $j$-subspace problem when $q=1$ is as follows. 

Suppose we have a $j$-subspace $C$ which is an $\alpha$-approximation of the optimal $j$-subspace for 
point set $P\in \REAL^d$. We want to find a $(j,k=1,\epsilon)$-coreset for point set $P$ in terms of an arbitrary 
$j$-subspace $L$. Here we denote the Euclidean distance of a point $p$ to $j$-subspace $L$ by $\dist(p,L)$. 
We also denote the orthogonal projection of a point $p$ onto $j$-subspace $C$ by $p'=proj(p,C)$. 
Let us fix a point $p\in P$ and its orthogonal projection $p'=proj(p,C)$ onto $C$. 
From elementary Calculus We have that $\dist(p,L)=\dist(p,L)+\dist(p',L)-\dist(p',L)$. 
We can rearrange this equation and show it as $\dist(p,L)=\dist(p',L)+(\dist(p,L)-\dist(p',L))$. 
By the triangle inequality we have that  $|\dist(p,L)-\dist(p',L)|\le \dist(p,p')$. 
Thus, a coreset for point $p$ would be to keep point $p'$ and try to approximate the term $(\dist(p,L)-\dist(p',L))$ 
using the random sampling. In this way, point set $P$ is replaced by $P'=\{p'=proj(p,C): p\in P\}$ and a set of pairs 
$(p,-p')$ for each $p\in P$. For pairs $(p,-p')$ we do random sampling to approximate the . Here we can use a similar technique as Chen's 
coreset technique \cite{Ch09} 
}




\subsection{Proof of Lemma \ref{lem:bucket:sum:small:error}}
\label{app:proof:bucket:sum:small:error}

\begin{proofof}{Lemma \ref{lem:bucket:sum:small:error}}
Recall that we take $(k,\epsilon)$-coreset $B_i$ from its children which are buckets $B_{i-1}$ and $B'_{i-1}$. 
Observe that $B_{i-1}$ and $B'_{i-1}$ are also $(k,\epsilon)$-coresets of subtrees rooted at nodes $B_{i-1}$ and $B'_{i-1}$
with partitions $\Lambda_{B_{i-1}}$ and $\Lambda_{B'_{i-1}}$, respectively.  
Similarly, we take $(k,\epsilon)$-coreset $B_{i-1}$ from its children which are buckets $B_{i-2}$ and $B'_{i-2}$ 
that  are, in turn,  $(k,\epsilon)$-coresets of subtrees rooted at nodes $B_{i-2}$ and $B'_{i-2}$.
We let $O_{i-1}=B_{i-2} \cup B'_{i-2}$. 
Let us fix arbitrary regions $R_{i-1}\in \Lambda_{B_{i-1}}$ and $R_{i}\in \Lambda_{B_i}$ such that $R_{i-1}\cap R_i \neq \emptyset$. 

Recall that we take a sample set of size $r_{i-1}=\min\big( |O_{i-1} \cap R_{i-1}|, \max\big(s_{\mathcal{CC}},O(d_{VC}\epsilon_{i-1}^{-2}\log(n)\cdot \log(\frac{d_{VC}\log(n)}{\epsilon_{i-1}\delta}))\big)\big)$  
points uniformly at random from $O_{i-1} \cap R_{i-1}$, assign a weight of 
$(|O_{i-1} \cap R_{i-1}|)/r_{i-1}$ to every sampled point, and we add the weighted sampled points to $B_{i-1}$. 
Here, $s_{\mathcal{CC}}$ is from Algorithm {\sc Unified}. 
This essentially means, $B_{i-1}\cap R_{i-1}$ is an $\epsilon_{i-1}$-sample for $O_{i-1} \cap R_{i-1}$.  
Since we treat a weighted point $p$ having weight $w_p$ as $w_p$ points at coordinates of $p$, 
we have $|B_{i-1} \cap R_{i-1}| = |O_{i-1} \cap R_{i-1}|$. 
Thus, 
$$
\big| | (O_{i-1} \cap R_{i-1}) \cap R_i|- | (B_{i-1} \cap R_{i-1}) \cap R_i| \big| 
\le \epsilon_{i-1} \cdot |O_{i-1} \cap R_{i-1}| \enspace . 
$$

Observe that  for regions $R_i$ and $R_{i-1}$, 
using Lemma \ref{lem:mix:vc:dim}, $(P,R_i \cap R_{i-1})$ is a range space of dimension $O(2d_{VC}\log(2d_{VC}))$. 
So, we can write 
$$
\big| | O_{i-1} \cap (R_{i-1} \cap R_i)|- | B_{i-1} \cap (R_{i-1} \cap R_i)| \big| 
\le \epsilon_{i-1} \cdot |O_{i-1} \cap R_{i-1}| \enspace . 
$$

Now let us expand $O_{i-1} \cap (R_{i-1} \cap R_i)$ where we use $O_{i-1}=B_{i-2} \cup B'_{i-2}$. 
Since $B_{i-2} $ and $ B'_{i-2}$ are disjoint, we then have 
$(B_{i-2} \cup B'_{i-2}) \cap (R_{i-1} \cap R_i)= (B_{i-2} \cap (R_{i-1} \cap R_i)) \cup (B'_{i-2}  \cap (R_{i-1} \cap R_i))$.

Similarly, let us consider $(k,\epsilon)$-coreset $B_{i-2}$ with its partition $\Lambda_{B_{i-2}}$ which 
is a $(k,\epsilon)$-coreset of its children, i.e., buckets $B_{i-3}$ and $B'_{i-3}$. 
Again, let $O_{i-2}=B_{i-3} \cup B'_{i-3}$. 
Let us fix an arbitrary region $R_{i-2}\in \Lambda_{B_{i-2}}$ such that $R_{i-2}\cap R_{i-1}\cap R_i \neq \emptyset$. 
Again, we take a sample set of size $r_{i-2}=\min\big( |O_{i-2} \cap R_{i-2}|, \max\big(s_{\mathcal{CC}},O(d_{VC}\epsilon_{i-2}^{-2}\log(n)\cdot \log(\frac{d_{VC}\log(n)}{\epsilon_{i-2}\delta}))\big)\big)$  
points uniformly at random from $O_{i-2} \cap R_{i-2}$, assign a weight of 
$(|O_{i-2} \cap R_{i-2}|)/r_{i-2}$ to every sampled point and we add the weighted sampled points to $B_{i-2}$. 
Since $B_{i-2}\cap R_{i-2}$ is an $\epsilon_{i-2}$-sample for $O_{i-2} \cap R_{i-2}$ and $|B_{i-2} \cap R_{i-2}| = |O_{i-2} \cap R_{i-2}|$, 
we then have  
\[
  \begin{split}
     \big| | (O_{i-2} \cap R_{i-2}) \cap (R_{i-1}\cap R_i)| - | (B_{i-2} \cap R_{i-2}) \cap (R_{i-2}\cap R_i)| \big| 
      \le \epsilon_{i-2}\cdot (|O_{i-2} \cap R_{i-2}|) \enspace . 
  \end{split}
\]

Observe that  for regions $R_i$,  $R_{i-1}$ and $R_{i-2}$, 
using Lemma \ref{lem:mix:vc:dim}, $(P,R_i \cap R_{i-1} \cap R_{i-2})$ is a range space of dimension $O(3d_{VC}\log(3d_{VC}))$. 
So, we can write 
\[
  \begin{split}
     \big| | O_{i-2} \cap (R_{i-2} \cap R_{i-1}\cap R_i)| - |B_{i-2} \cap(R_{i-2} \cap R_{i-2}\cap R_i)| \big| 
      \le \epsilon_{i-2}\cdot (|O_{i-2} \cap R_{i-2}|) \enspace . 
  \end{split}
\]

We do the same for $B'_{i-2}$. We define $O'_{i-2}$ for $B'_{i-2}$ similar to $O_{i-2}$. Using triangle inequality we have 
\[
  \begin{split}
     &\big| \sum_{R_{i-2}\in \Lambda_{B_{i-2}}} | O_{i-2} \cap (R_{i-2} \cap R_{i-1}\cap R_i)| + \sum_{R'_{i-2}\in \Lambda_{B'_{i-2}}} | O'_{i-2} \cap (R'_{i-2} \cap R_{i-1}\cap R_i)|  
     - | B_{i-1} \cap (R_{i-1} \cap R_i)| \big|\\
      &\le \epsilon_{i-2}\cdot \big(\sum_{R_{i-2}\in \Lambda_{B_{i-2}}}  (|O_{i-2} \cap R_{i-2}|) +\sum_{R'_{i-2}\in \Lambda_{B'_{i-2}}}  (|O'_{i-2} \cap R'_{i-2}|) \big) + \epsilon_{i-1} \cdot |O_{i-1} \cap R_{i-1}|  \enspace .
   \end{split}
\]

Now we recurse from level $i-2$ down to level $2$ in which we have $(k,\epsilon)$-coreset $B_{2}$ with partition $\Lambda_{B_{2}}$ 
which is a $(k,\epsilon)$-coreset of its children, i.e., buckets $B_{1}$ and $B'_{1}$. 
Again, let $O_{2}=B_{1} \cup B'_{1}$. 
Let us fix an arbitrary region $R_{2}\in \Lambda_{B_{2}}$ such that $R_{2}\cap \cdots \cap R_i \neq \emptyset$. 
Once again, we take a sample set of size $r_{2}=\min\big( |O_{2} \cap R_{2}|, \max\big(s_{\mathcal{CC}},O(d_{VC}\epsilon_{2}^{-2}\log(n)\cdot \log(\frac{d_{VC}\log(n)}{\epsilon_{2}\delta}))\big)\big)$ 
points uniformly at random from $O_{2} \cap R_{2}$, assign a weight of 
$\frac{|O_{2} \cap R_{2}|}{r_2}$ to every sampled point, and we add the weighted sampled points to $B_{2}$. 
Since $B_{2}\cap R_{2}$ is an $\epsilon_{2}$-sample for $O_{2} \cap R_{2}$ and $|B_{2} \cap R_{2}| = |O_{2} \cap R_{2}|$, we then have 
\[
  \begin{split}
     \big| | (O_{2} \cap R_{2}) \cap (R_3\cap \cdots \cap R_i)| - | (B_{2} \cap R_{2}) \cap (R_{3}\cap \cdots \cap R_i)| \big| 
      \le \epsilon_{2}\cdot (|O_{2} \cap R_{2}|) \enspace . 
  \end{split}
\]

Observe that for regions $R_i, R_{i-1}, \cdots R_{2}$ using 
Lemma \ref{lem:mix:vc:dim}, $(P,R_i \cap \cdots \cap R_{2})$ is a range space of dimension $O((i-1)d_{VC}\log((i-1)d_{VC}))$. 
So, we can write 
\[
  \begin{split}
     \big| | O_{2} \cap (R_{2} \cap \cdots \cap R_i)| - |B_{2} \cap(R_{2}\cap \cdots  \cap R_i)| \big| 
      \le \epsilon_{2}\cdot (|O_{2} \cap R_{2}|) \enspace . 
  \end{split}
\]

By repeated applications of the triangle inequality for levels $i-2$ down to $2$ we obtain 
\[
  \begin{split}
     &\big| \sum_{\text{node } x_2 \text{ in level } 2 } | O_{x_2} \cap (R_{i-1}\cap R_i)| - | B_{i-1} \cap (R_{i-1} \cap R_i)| \big|\\
     &=\big| \sum_{ x_2 \text{ in level } 2 } \sum_{R^{x_2}_{2}\in \Lambda_{x_2}} \sum_{R^{x_3}_{3}\in \Lambda_{x_3}} \cdots \sum_{R^{x_{i-2}}_{i-2}\in \Lambda_{x_{i-2}}}| O_{x_2} \cap (R^{x_2}_{2}\cap R^{x_3}_{3} \cap \cdots \cap R^{x_{i-2}}_{i-2})\cap R_{i-1}\cap R_i| - | B_{i-1} \cap (R_{i-1} \cap R_i)| \big|\\
      &\le \sum_{\text{level } j=2}^{i-2} \sum_{\text{node } x_j \text{ in level } j } \epsilon_{j}\cdot \big(\sum_{R\in \Lambda_{x_j}}  (|O_{x_j} \cap R|)  \big)  
      + \epsilon_{i-1} \cdot |O_{i-1} \cap R_{i-1}|  \enspace ,
   \end{split}
\]
where $x_2$ is a child of node $x_3$, $x_3$ is a child of node $x_4$, and so on, and $O_{x_j}$ is the point set at node $x_j$. 
Observe that in level $1$ (i.e., leaf level) we do not merge buckets and merging buckets starts at level $2$, because of that 
the index of the first sum start with $j=2$. We take sums $\sum_{R_{i-1}\in \Lambda_{B_{i-1}}}$ and $\sum_{R'_{i-1}\in \Lambda_{B'_{i-1}}}$
to conclude 
\[
   \begin{split}
     \big| |P_i \cap R_i| - |B_i \cap R_i| \big|
      \le \sum_{\text{level } j=2}^{i-1} \sum_{\text{node } x_j \text{ in level } j } \epsilon_{j} \big(\sum_{R\in \Lambda_{x_j}}  (|O_{x_j} \cap R|)  \big)  
      \enspace .  
    \end{split}
\]

In Algorithm \ref{alg:unify:coreset} we simply replace VC-dimension $d_{VC}$ by $O(d_{VC}\log n)$ for all levels $i\in \log n$. 
\end{proofof}

\end{document}